\documentclass[11pt,a4paper]{article}
\usepackage[margin=2.5cm,marginparwidth=2cm]{geometry}
\usepackage{lineno}
\usepackage{graphicx}
\usepackage{subcaption}
\usepackage{amsmath}
\usepackage{amsthm}               
\usepackage{enumitem}
\usepackage{newtxtext}
\usepackage{newtxmath}
\usepackage{natbib}
\usepackage{xcolor}
\usepackage{comment}
\usepackage[disable]{todonotes}   
\usepackage{hyperref}             
\hypersetup{
    colorlinks = true,
    urlcolor   = blue,
    citecolor  = black,
    pdftitle   = {Instability of exact finite-amplitude acoustic-gravity waves},
    pdfauthor  = {Christian Puntini}}
 
\newcommand{\e}{\mathfrak{e}}                 
\newcommand{\cD}{\mathcal{D}}                 
\newcommand{\bU}{\mathbf{U}}
\newcommand{\fU}{\boldsymbol{\mathfrak{U}}}   
\newcommand{\tM}{\widetilde{\mathfrak{M}}}    
\newcommand{\bxi}{\boldsymbol{\xi}}
\newcommand{\bn}{\mathbf{n}}
\newcommand{\bq}{\mathbf{q}}
\newcommand{\bp}{\mathbf{p}}
\newcommand{\bV}{\mathbf{V}}
\newcommand{\tV}{\widetilde{\mathbf{V}}}
\newcommand{\be}{\mathbf{e}}
\newcommand{\rh}{{\mathfrak{R}}}
\newcommand{\ech}{\hat{\e}_c}
 
\newtheorem{theorem}{Theorem}[section]
\newtheorem{proposition}[theorem]{Proposition}

\theoremstyle{remark}
\newtheorem{remark}[theorem]{Remark}
 
\title{Instability of exact finite-amplitude acoustic-gravity waves}
\author{Christian Puntini\thanks{Email: \href{mailto:christian.puntini@univie.ac.at}{christian.puntini@univie.ac.at}}\\[4pt]
  \small Faculty of Mathematics, University of Vienna\\
  \small Oskar-Morgenstern-Platz 1, 1090 Vienna, Austria}
\date{}
 
\begin{document}
\maketitle
 
\begin{abstract}
\noindent
Godin (\textit{J. Fluid Mech.}, vol. 767, 2015, pp. 52--64) showed that compressible fluids in a uniform gravity field support finite-amplitude incompressible wave motions, in which every parcel conserves pressure and density, and constructed exact Gerstner-type surface and edge waves, possibly carried by currents. Using the short-wavelength instability method for compressible flows, we show that these waves are linearly unstable above a critical steepness. The stability of each fluid level reduces to a Floquet problem depending on two parameters: the steepness $\e=e^{ka_2}$ and $\nu^2=(1-\e^2)\nu_0^2$, where $\nu_0^2=\cos^2\alpha\,N_0^2/\Omega^2$, $N_0$ is the buoyancy frequency of the hydrostatic reference state, $\alpha$ the inclination of the plane of motion and $\Omega=(gk\cos\alpha)^{1/2}$ the intrinsic orbital frequency. In an isentropic fluid, a level is unstable if $\e>1/3$, recovering Leblanc's threshold for Gerstner waves (\textit{J. Fluid Mech.}, vol. 506, 2004, pp. 245--254) and extending it to compressible surface and edge waves; weak stable stratification lowers the threshold to $\frac13-\frac{32}{81}\nu_0^2+O(\nu_0^4)$. In an isothermal atmosphere, $\nu_0^2=\kappa\cos\alpha/(kH)$, where $H$ is the pressure scale height. For horizontal wavelengths near $5.5\,H/\cos\alpha$, principal subharmonic resonance makes the critical steepness vanish. For the fast airglow wave observed by Hecht et al. (\textit{J. Geophys. Res.}, vol. 107, 2002, D20, 4444), which Godin interpreted with his solution, the lidar temperature places the layer in which this interpretation is consistent on this resonance, where the critical steepness is typically 0.1 instead of $1/3$. Instability is confined to a layer beneath the free surface of surface waves and to a wedge at the shoreline of edge waves.
 
\end{abstract}
 
\section{Introduction}\label{sec:intro}

\citet{G} showed that a compressible fluid in a uniform gravity field supports finite-amplitude wave motions that are {incompressible} in the sense that $\nabla\cdot\bU=0$, while pressure, density and entropy are conserved by each parcel: $DP/Dt=D\rho/Dt=DS/Dt=0$. The kinematics of these motions is that of Gerstner's wave; compressibility only shapes the distribution of pressure and density across the trochoidal isobars. Godin's family contains surface waves, edge waves along a sloping rigid wall, waves on uniform and piecewise-uniform currents, and waves carrying a transverse current.\\
\noindent
{To the best of our knowledge, the stability of these waves has not been studied. For Gerstner's wave in a homogeneous liquid, \citet{L04} showed with the short-wavelength method that a level is unstable when its steepness exceeds $1/3$, and \citet{IK14} extended the analysis to edge waves along a sloping beach. In \citet{P}, the exact mountain waves of \citet{C23} were analysed with the short-wavelength method of \citet{LH91} and \citet{LL92} for compressible flows. That analysis rests on three facts: the basic flow is divergence-free, $P$ and $\rho$ are constant along trajectories, and the particle paths are Gerstner trochoids. These are precisely the defining properties of Godin's class. Here we show that the stability of Godin's waves reduces exactly, level by level, to the Floquet problem of \citet{P}, and that the stratification enters through a parameter $\nu_0$ fixed by the hydrostatic reference state. Unlike the mountain waves of \citet{P}, whose stratification parameter is small, fast acoustic-gravity waves have $\nu_0$ of order $1/2$, where a parametric resonance makes the critical steepness collapse. The main results are the following. }
\begin{enumerate}[label=(\roman*),leftmargin=*]
\item \emph{Exact reduction.} At every level of the fluid, perturbations with wavevector normal to the plane of motion obey the amplitude system of \citet{P}, with steepness $\e=e^{ka_2}$ and $\nu^2=(1-\e^2)\nu_0^2$, $\nu_0^2=\cos^2\!\alpha\,N_0^2/\Omega^2$. Here $N_0$ is the buoyancy frequency of the hydrostatic reference state evaluated on the isobar, $\alpha$ is the angle between the upward vertical and the $x_3$-axis of the plane of motion ($\alpha=0$ for surface waves), and $\Omega=(gk\cos\alpha)^{1/2}$ is the intrinsic orbital frequency. Uniform currents drop out (Proposition~\ref{prop:equiv}, Theorem~\ref{thm:main}, Section~\ref{sec:currents}).
\item \emph{Compressible surface and edge waves.} In an isentropic fluid (homogeneous incompressible liquid, adiabatic atmosphere) a level is unstable if and only if $\e>1/3$, with an explicit growth rate; this recovers the threshold of \citet{L04} for Gerstner waves. A surface wave is therefore unstable for $H_w/\lambda>1/(3\pi)$, a height below that of the highest Stokes wave, and an edge wave, in a wedge at the shoreline, as soon as the shoreline oscillates by more than $\lambda\sin s/(6\pi)$ on a beach of slope $s$ (\S\S\,\ref{cor:isentropic} and~\ref{sec:where}).
\item \emph{Resonant collapse of the threshold.} Weak stable stratification lowers the threshold to $1/3-\frac{32}{81}\nu_0^2+O(\nu_0^4)$ (\S\,\ref{cor:weak}), and the threshold vanishes at the parametric resonances $\nu_0=n/2$ (\S\,\ref{sec:numerics}). In an isothermal atmosphere $\nu_0^2=\kappa\cos\alpha/(kH)$, so the whole wave field lies on one horizontal line of the stability diagram; the principal subharmonic resonance $\nu_0=1/2$ falls at horizontal wavelengths $\lambda\approx5.5\,H/\cos\alpha$, about 30--33\,km near the mesopause, in the range of the fast acoustic-gravity waves considered by \citet{G} (\S\,\ref{cor:poly}).
\item \emph{The wave of \citet{H02}.} For this fast airglow wave, which Godin interpreted with his solution, the lidar profiles measured during the observation show that the interpretation is consistent only between about 98 and 108\,km, above the airglow layer, and that the principal resonance lies inside this layer: there the critical steepness is typically $0.1$ instead of $1/3$, and a steepness of $0.2$ would give e-folding times of 6--9\,min. The unaveraged lidar profiles could decide whether the wave was steep enough (Section~\ref{sec:physics}).
\end{enumerate}
\noindent
The paper is organised as follows. Section~\ref{sec:godin} recalls the equations of motion and Godin's solutions.  Section~\ref{instability} adapts the short-wavelength method to these flows, reduces the problem to the amplitude system of \citet{P} and states the instability results for Godin's solution in \S \ref{sec:results}. Moreover, \S~\ref{sec:numerics} presents the numerical stability diagram, and \S~\ref{sec:where}  locates the unstable region in surface and edge waves. Section~\ref{sec:currents} discusses the role of currents and Section~\ref{sec:physics} the physical estimates, including the wave observed by \citet{H02}. Section~\ref{sec:discussion} summarises and discusses the results.

\section{Equations of motion and Godin's exact solutions}\label{sec:godin}
The motion of an inviscid compressible fluid in a uniform, time-independent gravity field $\mathbf{g}$ is governed by the following set of (partial differential) equations
\begin{align}
    D\mathbf{U}/Dt=\mathbf{g} - \rho^{-1}\nabla P,  \label{2.1}\\
    D\rho/Dt + \rho\nabla\cdot\mathbf{U}=0,\label{2.2}\\ 
    P=P(\rho,S). \label{2.3}
\end{align}
In \eqref{2.1}, \eqref{2.2}, and \eqref{2.3}, $\mathbf{U}=(U_1,U_2,U_3)$ denotes the fluid velocity, $P$ the pressure, $\rho$ the fluid density, and $S$ its entropy, all regarded as functions of the spatial coordinates $\mathbf{x}=(x_1,x_2,x_3)$. Equations \eqref{2.1} and \eqref{2.2} express, respectively, the conservation of momentum and mass, while \eqref{2.3} represents the equation of state, all formulated in the Eulerian framework.\\
Moreover, the gravity is $\mathbf g=(0,-g\sin\alpha,-g\cos\alpha)$ with $0\le\alpha<\pi/2$ (the angle $\alpha$ is called $\theta$ in \citet{G}, and we reserve $\theta$ for the orbital phase of \citet{P}). For $\alpha=0$ the $x_3$-axis is vertical and the waves are surface waves. For $\alpha>0$ the plane $x_2=0$ is a rigid wall inclined at $\pi/2-\alpha$ to the horizontal, the fluid occupies $x_2\ge0$, and the waves are edge waves; the coordinates are shown in Figure \ref{fig:entrambe}. In Lagrangian labels $\mathbf a:=(a_1,a_2,a_3)$ the particle paths of \citet[(4.1), (4.10)]{G} read
\begin{equation}\label{eq:paths}
\left\{\begin{aligned}
    x_1&=a_1+vt+\frac1k e^{ka_2}\sin\varphi,\\ x_2&=u(a_1,a_2)t-a_3,\\ x_3&=a_2-\frac1k e^{ka_2}\cos\varphi,
\end{aligned}\right.\qquad\text{where}\quad \varphi=k(a_1+c_*t),
\end{equation}
with $c_*:=c+v$ and $(a_1,\,a_2,\,a_3)\in(-\infty, +\infty)\times(-\infty,\, a_2^{(0)})\times(-\infty, +\infty)$, with $a_2^{(0)}<0$. Here $v$ and $u$ are the currents along and across the direction of propagation, and the dispersion relation is
\begin{equation}\label{eq:disp}
k\,c_*^2=g\cos\alpha,\qquad \Omega:=k|c_*|=\sqrt{gk\cos\alpha}.
\end{equation}
\begin{figure}
    \centering

    \begin{subfigure}{0.55\textwidth}
        \centering
        \includegraphics[width=\linewidth]{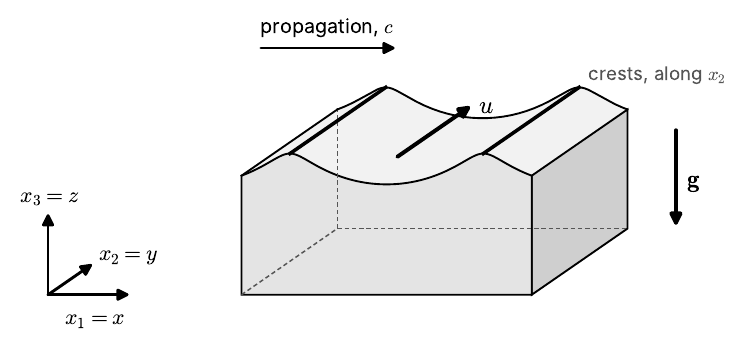}
       \caption{\empty}
        \label{fig:prima}
    \end{subfigure}
    \hfill
    \begin{subfigure}{0.4\textwidth}
        \centering
        \includegraphics[width=\linewidth]{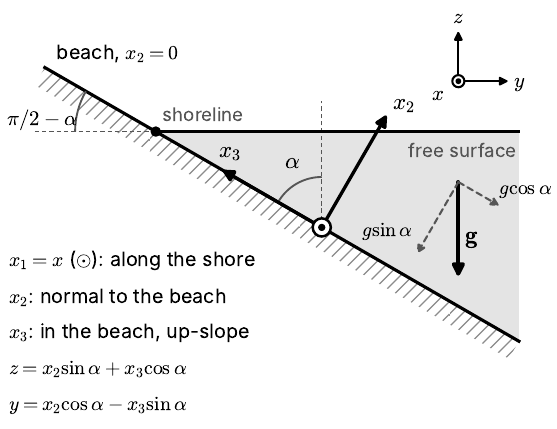}
       \caption{\empty}
        \label{fig:seconda}
    \end{subfigure}

    \caption{Coordinates of Godin's waves \eqref{eq:paths}. (\textit{a}) Surface waves ($\alpha=0$): the waves propagate along $x_1$, their crests are parallel to $x_2$, which is also the direction of the transverse current $u$, and $x_3$ is vertical. (\textit{b}) Edge waves ($\alpha>0$), in a section normal to the shoreline: $x_1$ is along the shore, $x_2$ is normal to the beach $x_2=0$ and points into the fluid, and $x_3$ lies in the beach and points up-slope. The beach is inclined at $\pi/2-\alpha$ to the horizontal, and $\mathbf g=(0,-g\sin\alpha,-g\cos\alpha)$.}
    \label{fig:entrambe}
\end{figure}
\noindent
The analysis carried out by \cite{G} considered two cases:
\begin{enumerate}
    \item\ waves in a quiescent fluid: in this case $u=v=0$, so that $c_*=c$ and the dispersion relation \eqref{eq:disp} reads as $kc^2=g\cos\alpha$. For $\alpha\not=0$ and $a_3\leq0$, \eqref{eq:paths} describes edge waves along the sloping rigid boundary given by $x_2=0$ (see also \citealp{Yih, C2001}).\\
    If $\alpha=0$, the particle paths in \eqref{eq:paths} characterise a surface wave in a fluid half-space (with surface at $a_2=a_2^{(0)}$) if $-\infty<a_3<+\infty$, or a wave that propagates along a rigid vertical boundary when $-\infty<a_3\leq0$.
    \item\ waves in a background flow: set $\alpha=0$ and $u$ and/or $v$ not zero, and assume $-\infty<a_{1,\,3}<+\infty$ and $-\infty<a_2\leq a_2^{(0)}<0$. 
\end{enumerate} 
Writing $\e:=e^{ka_2}$ for the steepness of the level $a_2$ and $\cD:=1-\e^2$, the velocity associated with \eqref{eq:paths} is
\begin{equation}\label{eq:velocity}
\bU=\bigl(v+\e c_*\cos\varphi,\ u,\ \e c_*\sin\varphi\bigr);
\end{equation}
it depends on $(x_1,x_3,t)$ only, never on $x_2$. The Jacobian matrix of \eqref{eq:paths},
with entries $(\partial x_i/\partial a_j)$, is given by
\begin{equation}\label{eq:jac}
\left[\frac{\partial(x_1,x_2,x_3)}{\partial(a_1,a_2,a_3)}\right]=
\begin{pmatrix}
1+\e\cos\varphi & \e\sin\varphi & 0\\[2pt]
t\,\partial_{a_1}u & t\,\partial_{a_2}u & -1\\[2pt]
\e\sin\varphi & 1-\e\cos\varphi & 0
\end{pmatrix},
\end{equation}
with determinant
\begin{equation}\label{eq:det}
\det\left[\frac{\partial(x_1,x_2,x_3)}{\partial(a_1,a_2,a_3)}\right]
=1-\e^2=\cD>0,
\end{equation}
which depends neither on $t$ nor on $u$; hence $\nabla\cdot\boldsymbol \bU=0$.\\
The inverse matrix, with entries $(\partial a_i/\partial x_j)$, is
\begin{equation}\label{eq:jacinv}
\begin{aligned}&\left[\frac{\partial(a_1,a_2,a_3)}{\partial(x_1,x_2,x_3)}\right]=\\
&=\frac{1}{\cD}
\begin{pmatrix}
1-\e\cos\varphi & 0 & -\e\sin\varphi\\[2pt]
-\e\sin\varphi & 0 & 1+\e\cos\varphi\\[2pt]
t\bigl[(1-\e\cos\varphi)\,\partial_{a_1}u-\e\sin\varphi\,\partial_{a_2}u\bigr] & -\cD &
t\bigl[(1+\e\cos\varphi)\,\partial_{a_2}u-\e\sin\varphi\,\partial_{a_1}u\bigr]
\end{pmatrix}.\end{aligned}
\end{equation}

\noindent
For constant $u$ the second row of \eqref{eq:jac} reduces to $(0,\,0,\,-1)$ and $\nabla a_3=(0,\,-1,\,0)$.\\

\noindent
The pressure and density distributions are obtained by substituting \eqref{eq:paths} into the Lagrangian Euler equations and using \eqref{eq:disp}, giving
\begin{equation}\label{eq:pa}
\frac{\partial P}{\partial a_1}=0,\qquad
\frac{\partial P}{\partial a_2}=-(1-e^{2ka_2})\,\rho g\cos\alpha,\qquad
\frac{\partial P}{\partial a_3}=\rho g\sin\alpha,
\end{equation}
which are solved by $P=g f(b)$, $\rho=-f'(b)$, with $f$ decreasing and
\begin{equation}\label{eq:b}
b=\Bigl(a_2-\frac{e^{2ka_2}}{2k}\Bigr)\cos\alpha-a_3\sin\alpha .
\end{equation}

\begin{remark}\label{rem:misprint}
Equation (4.4c) of \citet{G} prints $b=a_2(1-e^{2ka_2})\cos\alpha-a_3\sin\alpha$. Its $a_2$-derivative is $(1-e^{2ka_2}-2ka_2e^{2ka_2})\cos\alpha$, so it does not satisfy the second relation in \eqref{eq:pa}; the expression \eqref{eq:b} does, and agrees with the pressure of \citet{C23}. Accordingly, for $\alpha=0$ the free surface $a_2=a_2^{(0)}$ corresponds to $b_0=a_2^{(0)}-e^{2ka_2^{(0)}}/(2k)<-1/(2k)$. Only $\partial b/\partial a_2=(1-e^{2ka_2})\cos\alpha$ and $\partial b/\partial a_3=-\sin\alpha$ enter below.
\end{remark}

\subsection{Thermodynamics through the reference state}
In hydrostatic equilibrium the same construction gives $P=gf(B)$, $\rho=-f'(B)$, with $B=x_2\sin\alpha+x_3\cos\alpha$ the height along $-\mathbf g$ \cite[(4.7)]{G}; we call this the reference state, with profiles $P_s(B)$ and $\rho_s(B)$. In the wave, $P=P_s(b)$ and $\rho=\rho_s(b)$: each particle carries the pressure and density that the reference state has at the height $B=b$, and the isobars, which are also isopycnals and material surfaces, are the trochoidal surfaces $b=\mathrm{const}$; likewise $C=C(b)$ (figure~\ref{fig:reference}). 
\begin{figure}
\centering
\includegraphics[width=.8\textwidth]{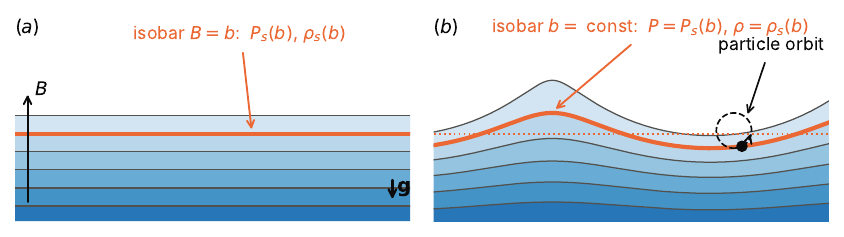}
\caption{Godin's construction of pressure and density, for $\alpha=0$. (\textit{a}) The reference state, the hydrostatic equilibrium with the same function $f$: the isobars are the horizontal planes $B=\mathrm{const}$, with pressure $P_s(B)$ and density $\rho_s(B)$. (\textit{b}) The wave: each isobar $b=\mathrm{const}$ is a trochoidal material surface, which is also an isopycnal, and carries the values $P_s(b)$ and $\rho_s(b)$ that the reference state has at the height $B=b$; the uppermost isobar is the free surface. Particles move on circles (dashed) without leaving their isobar, and the mean height of the isobar $b$ is $b$ (dotted).}
\label{fig:reference}
\end{figure}
We measure the stratification by the buoyancy frequency of the reference state at the height $b$, e.g. \citep{Gill82},
\begin{equation}\label{eq:N0}
N_0^2(b):=-g\left(\frac{\rho_s'(b)}{\rho_s(b)}+\frac{g}{C^2(b)}\right),
\end{equation}
which depends only on the reference state; the deformation of the isobars by the wave enters the analysis separately, through $\nabla b$ (Section~\ref{sec:reduction}).
For an ideal gas with constant $\gamma$ and a polytropic reference state $P_s\propto\rho_s^{\,n}$ \cite[(4.8)--(4.9)]{G}, the hydrostatic balance $P_s'=-g\rho_s$ and $C^2=\gamma P_s/\rho_s$ give
\begin{equation}\label{eq:N0poly}
N_0^2=\frac{g}{H}\Bigl(\frac1n-\frac1\gamma\Bigr),\qquad H:=\frac{P_s}{\rho_s g}=\frac{R_dT}{g},\end{equation}
where $R_d\approx287$\,J\,kg$^{-1}$\,K$^{-1}$ is the specific gas constant of air and $T$ the temperature of the reference state, so $N_0\equiv0$ in an adiabatic atmosphere ($n=\gamma$) and $N_0^2=\kappa g/H$, $\kappa=(\gamma-1)/\gamma$, in an isothermal one ($n=1$) \citep{H60,GH75}. A homogeneous incompressible liquid, and more generally any barotropic fluid $P=P(\rho)$, also has $N_0\equiv0$.

The analysis below uses three exact properties only: $\nabla\cdot\bU=0$; $P$, $\rho$ and $C$ are functions of $b$, hence constant along particle paths; and $\bU$ does not depend on $x_2$.


\section{Instability analysis}\label{instability}
\subsection{The short-wavelength instability method}\label{sec:wkb}

Throughout, $(\bU,P,\rho)$ denotes one of the basic flows of Section~\ref{sec:godin}. The motion is adiabatic, $DS/Dt=0$, so the equation of state \eqref{2.3} gives $DP/Dt=C^2\,D\rho/Dt$, with $C^2:=(\partial P/\partial\rho)_S$ the squared sound speed, and \eqref{2.2} gives $DP/Dt=-\rho C^2\,\nabla\cdot\bU$. Since $\nabla\cdot\bU=0$ and $C=C(P,S)$,
\begin{equation}\label{eq:basicDt}
\frac{DP}{Dt}=\frac{D\rho}{Dt}=\frac{DS}{Dt}=\frac{DC}{Dt}=0 .
\end{equation}

\subsubsection{Linearised equations}
We superimpose on the basic flow small perturbations of velocity, density and pressure,
\begin{equation*}
\bU\to\bU+\mathfrak u,\qquad \rho\to\rho+\mathfrak{r},\qquad P\to P+\mathfrak p,\qquad \mathfrak u=(\mathfrak u_1,\mathfrak u_2,\mathfrak u_3).
\end{equation*}
Substituting into \eqref{2.1}, \eqref{2.2} and into the adiabatic relation $DP/Dt=-\rho C^2\,\nabla\cdot\bU$, expanding $1/(\rho+\mathfrak{r})=1/\rho-\mathfrak{r}/\rho^2+O(\mathfrak{r}^2)$ and keeping only the terms linear in the perturbations, we obtain
\begin{equation}\label{eq:lin}
\left\{\begin{aligned}
&\frac{D\mathfrak u}{Dt}+(\nabla\bU)\,\mathfrak u+\frac{\nabla\mathfrak p}{\rho}-\frac{\mathfrak{r}}{\rho^2}\,\nabla P=0,\\
&\frac{D\mathfrak{r}}{Dt}+\mathfrak u\cdot\nabla\rho+\rho\,\nabla\cdot\mathfrak u=0,\\
&\frac{D\mathfrak p}{Dt}+\mathfrak u\cdot\nabla P+\rho C^2\,\nabla\cdot\mathfrak u=0,
\end{aligned}\right.
\end{equation}
where $D/Dt=\partial_t+\bU\cdot\nabla$ is the material derivative along the basic flow, $(\nabla\bU)_{ij}=\partial U_i/\partial x_j$, and the terms proportional to $\nabla\cdot\bU$, among which the one containing the perturbation of $\rho C^2$, have been dropped. The system is completed by the initial data $(\mathfrak u,\mathfrak{r},\mathfrak p)(\mathbf x,0)$.

Unlike \citet{P}, where the Eckart transform was used, we keep $\mathfrak{r}$ and $\mathfrak p$ as unknowns. By \eqref{eq:basicDt}, the last two equations in \eqref{eq:lin} combine into the exact relation
\begin{equation}\label{eq:identity}
\frac{D}{Dt}\Bigl(\mathfrak{r}-\frac{\mathfrak p}{C^2}\Bigr)=-\,\mathfrak u\cdot\Bigl(\nabla\rho-\frac{\nabla P}{C^2}\Bigr),
\end{equation}
in which $\nabla\cdot\mathfrak u$ no longer appears. By the equation of state, $\mathfrak{r}-\mathfrak p/C^2=(\partial\rho/\partial S)_P\,\delta S$, with $\delta S$ the entropy perturbation, so \eqref{eq:identity} is the linearised form of $DS/Dt=0$.

\subsubsection{Rapidly oscillating perturbations}
Following \citet{LH91} and \citet{LL92}, we consider localised perturbations of short wavelength in the following Wentzel-Kramers-Brillouin (WKB) form,
\begin{equation}\label{eq:ansatz}
(\mathfrak u,\mathfrak{r},\mathfrak p)(\mathbf x,t)=e^{\mathrm i\Phi(\mathbf x,t)/\varepsilon}\Bigl[\bigl(\fU,\rh,\varepsilon\mathfrak P\bigr)+\varepsilon\bigl(\fU_1,\rh_1,\varepsilon\mathfrak P_1\bigr)+O(\varepsilon^2)\Bigr],\qquad 0<\varepsilon\ll1,
\end{equation}
with a real phase $\Phi$, amplitudes depending on $(\mathbf x,t)$, and local wavevector $\bxi:=\nabla\Phi$. The pressure amplitude is taken one order smaller than the others: a non-zero pressure amplitude of order one is compatible only with the acoustic branch $(D\Phi/Dt)^2=C^2|\bxi|^2$, whereas we are interested in perturbations whose phase is advected by the flow (namely $D\Phi/Dt=0$).\\
At order $\varepsilon^{-1}$, the three equations \eqref{eq:lin} give, respectively, 
\begin{equation}
    \left(\frac{D\Phi}{Dt}\right)\,\fU=0, \qquad  \left(\frac{D\Phi}{Dt}\right)\,\rh+\rho\,\bxi\cdot\fU=0,\qquad \rho C^2\,\bxi\cdot\fU=0.
\end{equation}
Hence, for a non-trivial amplitude, the velocity amplitude is transverse to the wavevector and the phase is advected by the basic flow:
\begin{equation}\label{eq:eikonal}
\bxi\cdot\fU=0,\qquad \frac{D\Phi}{Dt}=0,\qquad\text{so that}\qquad \frac{D\bxi}{Dt}=-(\nabla\bU)^T\bxi .
\end{equation}
At order $\varepsilon^{0}$, using \eqref{eq:eikonal} we get
\begin{equation}\label{eq:order0}
\left\{\begin{aligned}
&\frac{D\fU}{Dt}+(\nabla\bU)\fU+\frac{\mathrm i\,\mathfrak P}{\rho}\,\bxi-\frac{\rh}{\rho^2}\,\nabla P=0,\\
&\frac{D\rh}{Dt}+\fU\cdot\nabla\rho+\rho\,\bigl(\nabla\cdot\fU+\mathrm i\,\bxi\cdot\fU_1\bigr)=0,\\
&\fU\cdot\nabla P+\rho C^2\bigl(\nabla\cdot\fU+\mathrm i\,\bxi\cdot\fU_1\bigr)=0 .
\end{aligned}\right.
\end{equation}
Eliminating $\nabla\cdot\fU+\mathrm i\,\bxi\cdot\fU_1$ between the last two equations gives $D\rh/Dt=-\fU\cdot(\nabla\rho-C^{-2}\nabla P)$, the leading-order form of \eqref{eq:identity}, while the last equation determines $\bxi\cdot\fU_1$. The pressure amplitude follows from the transversality constraint: differentiating $\bxi\cdot\fU=0$ along the flow and using \eqref{eq:eikonal} and the first equation of \eqref{eq:order0},
\begin{equation*}
\frac{\mathrm i\,\mathfrak P}{\rho}=\frac{1}{|\bxi|^2}\Bigl(\frac{\rh}{\rho^2}\,\bxi\cdot\nabla P-2\,\bxi\cdot(\nabla\bU)\fU\Bigr).
\end{equation*}
Substituting back, along each trajectory $\mathbf X(t)$ of the basic flow the wavevector and the leading-order amplitudes obey the system of ordinary differential equations
\begin{equation}\label{eq:wkb}
\left\{\begin{aligned}
\dot{\mathbf X}&=\bU, \\ \dot\bxi&=-(\nabla\bU)^{T}\bxi,\\
\dot\fU&=-(\nabla\bU)\fU+\frac{2\,\bxi\cdot(\nabla\bU)\fU}{|\bxi|^2}\,\bxi+\frac{\rh}{\rho^2}\Bigl(\mathbb I-\frac{\bxi\bxi^{T}}{|\bxi|^2}\Bigr)\nabla P,\\
\dot{\rh}&=-\fU\cdot\Bigl(\nabla\rho-\frac{\nabla P}{C^2}\Bigr),
\end{aligned}\right.
\end{equation}
where the dot denotes $D/Dt$ and $\mathbb I$ is the identity matrix; the constraint $\fU\cdot\bxi=0$ is preserved. System \eqref{eq:wkb} is completed by the initial data
\begin{equation*}
\mathbf X(0)=\mathbf X_0,\quad \bxi(0)=\bxi_0=\nabla\Phi_0(\mathbf X_0),\quad \fU(0)=\fU_0,\quad \rh(0)=\rh_0,\quad \bxi_0\cdot\fU_0=0 .
\end{equation*}

\begin{figure}
    \centering
    \includegraphics[width=.8\linewidth]{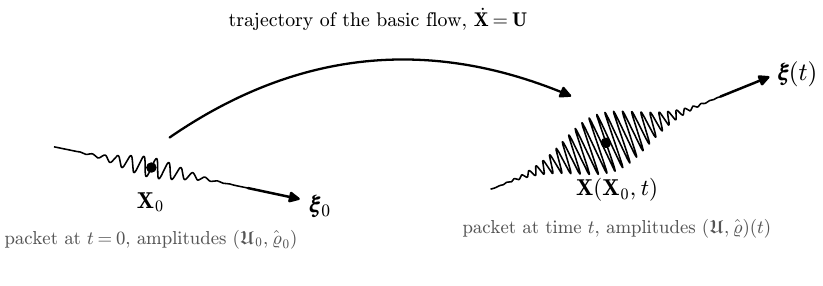}
    \caption{Sketch of the short-wavelength method. At $t=0$ a disturbance of the form \eqref{eq:ansatz} is localised around $\mathbf X_0$, with wavevector $\bxi(0)=\bxi_0$ and amplitudes $(\fU_0,\rh_0)$. It is advected by the basic flow and, at time $t$, is localised around $\mathbf X(\mathbf X_0,t)$, with wavevector $\bxi(t)$ and amplitudes $(\fU,\rh)(t)$ given by \eqref{eq:wkb}. The basic flow is unstable if, for some initial data, the amplitudes grow without bound, in the sense of \eqref{eq:instab}.}
    \label{fig:instability}
\end{figure}
\noindent
The initial velocity amplitude is taken of the form $\fU_0=\psi^{-1}\,\nabla F\times\nabla\Phi_0$, with $F$ smooth and
\begin{equation*}
\psi(P):=\exp\int^{P}\frac{\mathrm dp}{\rho C^2(p)} ,
\end{equation*}
which is well defined for Godin's flows, where $\rho$ and $C$ are functions of $b$ only and $P$ is a strictly monotone function of $b$; for an ideal gas with constant $\gamma$, $\rho C^2=\gamma P$ and $\psi=P^{1/\gamma}$, as in \citet{P}. Then $\bxi_0\cdot\fU_0=0$ and $\nabla\cdot(\psi\fU_0)=0$. Since the last equation of \eqref{eq:order0} can be written $\nabla\cdot(\psi\fU)=-\mathrm i\,\psi\,\bxi\cdot\fU_1$, the second condition makes these initial data, for which $\fU_1(0)=0$, consistent with \eqref{eq:ansatz} at order $\varepsilon$. It can be proved (see \citealp{LL92}) that the solution of \eqref{eq:lin} is then of the form \eqref{eq:ansatz}, up to terms of order $\varepsilon$.\\
\noindent
The system \eqref{eq:wkb} describes a sharply peaked disturbance, located at $t=0$ around $\mathbf X_0$, which is advected by the basic flow while its wavevector and amplitudes evolve (Figure \ref{fig:instability}). Following \citet{LH91} and \citet{LL92}, the basic flow is linearly unstable if
\begin{equation}\label{eq:instab}
\sup_{\substack{\mathbf X_0,\,\bxi_0,\,\fU_0,\,\rh_0\\ \bxi_0\cdot\fU_0=0}}\ \limsup_{t\to\infty}\ \bigl|\bigl(\fU,\rh\bigr)(\mathbf X_0,\bxi_0,\fU_0,\rh_0;t)\bigr|=\infty .
\end{equation}

\begin{remark}\label{rem:relP}\label{rem:eos}
In \citet{P} the Eckart variables $\mathfrak m$, $\mathfrak n$, defined by $\mathfrak{r}=\rho(\mathfrak m+\mathfrak n)/C$ and $\mathfrak p=\rho C\,\mathfrak n$, were used. In \eqref{eq:ansatz} the pressure is of order $\varepsilon$, so $\mathfrak n$ vanishes at leading order and the entropy amplitude of \citet{P} is $\mathfrak M=C\rh/\rho$. With this identification, and for an ideal gas with constant $\gamma$, system \eqref{eq:wkb} coincides with system (17) of \citet{P}. The derivation above uses no ideal-gas relation; it therefore applies to Godin's general fluids, including an ideal gas with $\gamma=\gamma(b)$ \citep[(4.5)]{G} and sea water.
\end{remark}

\subsection{Reduction of the system and Floquet analysis}\label{sec:reduction}
Having presented the general theory of short-wavelength instabilities for compressible flows, without the use of Eckart variables (as is instead the case  in \citet{LH91,LL92,P}), we now turn to the analysis of the wave motion of \cite{G}.\\
Throughout this section, we  consider the case of constants $u$ and $v$, and we postpone to \S\ref{sec:currents} the discussion of the case in which $u$ and $v$ are not constant. Moreover, we denote by $\be_2=(0,1,0)$.\\
We start by computing the gradients of the basic state that enter \eqref{eq:wkb}. Since $x_1$ and $x_3$ in \eqref{eq:paths} depend only on $(a_1,a_2,t)$ and $u$ is constant, the velocity \eqref{eq:velocity} does not depend on $x_2$ and its transverse component $U_2=u$ is uniform: $\partial\bU/\partial x_2=0$ and $\nabla U_2=0$.\\ The in-plane Jacobian matrix and its inverse are
\begin{equation}\label{eq:jacperp}
\begin{aligned}
\left[\frac{\partial(x_1,x_3)}{\partial(a_1,a_2)}\right]&=
\begin{pmatrix}
1+\e\cos\varphi & \e\sin\varphi\\[2pt]
\e\sin\varphi & 1-\e\cos\varphi
\end{pmatrix},\\
\left[\frac{\partial(a_1,a_2)}{\partial(x_1,x_3)}\right]&=
\begin{pmatrix}
\dfrac{1-\e\cos\varphi}{\cD} & -\dfrac{\e\sin\varphi}{\cD}\\[2pt]
-\dfrac{\e\sin\varphi}{\cD} & \dfrac{1+\e\cos\varphi}{\cD}
\end{pmatrix};\end{aligned}
\end{equation}
the rows of the inverse are the gradients of $a_1$ and $a_2$ in the $(x_1,x_3)$-plane. \\
The derivatives of $(U_1,U_3)$ with respect to $(a_1,a_2)$ are $kc_*\e$ times the rows $(-\sin\varphi,\ \cos\varphi)$ and $(\cos\varphi,\ \sin\varphi)$. Multiplying them by the inverse Jacobian and introducing
\begin{equation}\label{eq:Sn}
S(\varphi)=\begin{pmatrix}-\sin\varphi&\cos\varphi+\e\\ \cos\varphi-\e&\sin\varphi\end{pmatrix},\qquad
\bn(\varphi)=\begin{pmatrix}-\e\sin\varphi\\ 1+\e\cos\varphi\end{pmatrix},
\end{equation}
we obtain
\begin{equation}\label{eq:grads}
\nabla_{\!\perp}(U_1,U_3)=\frac{kc_*\e}{\cD}\,S(\varphi),\qquad
\nabla_{\!\perp}a_2=\frac{\bn(\varphi)}{\cD},
\end{equation}
where $\nabla_{\!\perp}$ denotes the gradient in $(x_1,x_3)$. Since $a_3=ut-x_2$, we have $\nabla a_3=-\be_2$, and \eqref{eq:b} gives
\begin{equation}\label{eq:gradb}
\nabla b=\frac{\partial b}{\partial a_2}\,\nabla a_2+\frac{\partial b}{\partial a_3}\,\nabla a_3=\cos\alpha\,(n_1,0,n_2)+\sin\alpha\,\be_2 .
\end{equation}
Finally, $P=gf(b)$ and $\rho=-f'(b)=\rho_s(b)$, so that, by \eqref{eq:N0},
\begin{equation}\label{eq:gradPrho}
\nabla P=-g\rho\,\nabla b,\qquad \nabla\rho-\frac{\nabla P}{C^2}=-\frac{\rho N_0^2}{g}\,\nabla b .
\end{equation}
\noindent
We consider perturbations whose wavevector is normal to the plane of motion, $\bxi(0)=\be_2$. Since $(\nabla\bU)^T\be_2=\nabla U_2=0$, the eikonal equation in \eqref{eq:wkb} gives $\dot\bxi=0$, and the constraint $\fU\cdot\bxi=0$ then forces the transverse amplitude to vanish:
\begin{equation}\label{eq:xi_e2}
\bxi(t)=\be_2,\qquad \mathfrak U_2\equiv0 .
\end{equation}
In the amplitude equation of \eqref{eq:wkb}, the term proportional to $\bxi$ vanishes, because $\bxi\cdot(\nabla\bU)\fU=\nabla U_2\cdot\fU=0$, and the projection $\mathbb I-\be_2\be_2^T$ removes the $\be_2$-component of $\nabla P$; consistently, $\dot{\mathfrak U}_2=0$. By \eqref{eq:gradb} and \eqref{eq:gradPrho}, the projected pressure gradient is $(\mathbb I-\be_2\be_2^T)\nabla P=-g\rho\cos\alpha\,(n_1,0,n_2)$, and the density equation gives $$\dot{\rh}=(\rho N_0^2/g)\,\fU\cdot\nabla b=(\rho N_0^2/g)\cos\alpha\,\bn\cdot\bV,\qquad \text{where}\qquad \bV:=(\mathfrak U_1,\mathfrak U_3).$$ Since $\rho$ is constant along trajectories, $m:=\rh/\rho$ satisfies $\dot m=\dot{\rh}/\rho$, and, using \eqref{eq:grads}, the in-plane amplitudes obey
\begin{equation}\label{eq:sys_t}
\dot\bV=-\frac{kc_*\e}{\cD}\,S(\varphi)\,\bV-g\cos\alpha\;m\,\bn(\varphi),\qquad
\dot m=\frac{N_0^2\cos\alpha}{g}\;\bn(\varphi)\cdot\bV .
\end{equation}

\begin{remark}    
The along-wall component $g\sin\alpha$ of gravity, which has no counterpart in \citet{P}, is absorbed by the projection: for edge waves it is balanced by the perturbation pressure, and only the effective gravity $g\cos\alpha$ acts in the plane of motion.
\end{remark}
\begin{proposition}[Equivalence with \citet{P}]\label{prop:equiv}
Set
\begin{equation}\label{eq:change}
\begin{gathered}
\theta:=\pi-\varphi,\qquad \tV:=R\,\bV,\quad R=\begin{pmatrix}
    -1 & 0\\
    0 & 1
\end{pmatrix},\qquad \tM:=c_*\,m,\\
\nu_0^2:=\frac{\cos^2\!\alpha\,N_0^2(b)}{\Omega^2},\qquad \nu^2:=(1-\e^2)\,\nu_0^2=\cD\,\nu_0^2  .
\end{gathered}
\end{equation}
Then, with ${}'=d/d\theta$ and
\begin{equation}\label{eq:Gq}
G(\theta)=\frac{\e}{\cD}\begin{pmatrix}-\sin\theta&\cos\theta-\e\\ \cos\theta+\e&\sin\theta\end{pmatrix},\qquad
\bq(\theta)=\begin{pmatrix}\e\sin\theta\\ 1-\e\cos\theta\end{pmatrix},
\end{equation}
equations \eqref{eq:sys_t} read as
\begin{equation}\label{eq:sysP}
\tV'=G(\theta)\,\tV+\tM\,\bq(\theta),\qquad \tM'=-\frac{\nu^2}{\cD}\,\bq(\theta)\cdot\tV,
\end{equation}
which is system (25) of \citet{P}. In particular, the monodromy matrix, the Floquet multipliers and the discriminant $\Delta$ are those of \citet{P} evaluated at $(\e,\nu^2)$ with $\nu^2=(1-\e^2)\nu_0^2$; they are described below.
\end{proposition}

\begin{proof}
Along a path $\dot\varphi=kc_*$, so $d/d\varphi=(kc_*)^{-1}d/dt$; this holds for either sign of $c_*$, and the factor $kc_*$ cancels in the velocity block. Using $g\cos\alpha=kc_*^2$ and $gk=\Omega^2/\cos\alpha$, \eqref{eq:sys_t} becomes
\begin{equation}\label{eq:sys_phi}
\frac{d\bV}{d\varphi}=-\frac{\e}{\cD}\,S\,\bV-\tM\,\bn,\qquad
\frac{d\tM}{d\varphi}=\frac{N_0^2\cos\alpha}{gk}\,\bn\cdot\bV=\nu_0^2\,\bn\cdot\bV .
\end{equation}
For $\varphi=\pi-\theta$ one has $\sin\varphi=\sin\theta$ and $\cos\varphi=-\cos\theta$, hence $R\,S(\pi-\theta)\,R=(\cD/\e)\,G(\theta)$ and $R\,\bn(\pi-\theta)=\bq(\theta)=(\e\sin\theta,\,1-\e\cos\theta)^T$. \\
Multiplying the first equation by $R$ and using $d/d\theta=-d/d\varphi$ and $(R\bn)\cdot(R\bV)=\bn\cdot\bV$ gives \eqref{eq:sysP}, since $\nu_0^2=\nu^2/\cD$.
\end{proof}



\noindent
By Proposition~\ref{prop:equiv}, the stability of each level of Godin's waves is decided by the linear system \eqref{eq:sysP}, whose coefficients are $2\pi$-periodic in $\theta$ and depend on the two parameters $(\e,\nu^2)\in(0,1)\times\mathbb R$. This system was analysed by \citet[\S3.1]{P}; we collect here, without proofs, the results used below, restated in the present notation. \\
Denoting by $\mathbf E:=(\widetilde V_1,\widetilde V_2,\tM)^T$, the system \eqref{eq:sysP} reads as $\mathbf E'=\mathbb A(\theta)\,\mathbf E$, where
\begin{equation}\label{eq:A}
\mathbb A(\theta)=
\begin{pmatrix}
-\dfrac{\e\sin\theta}{\cD} & \dfrac{\e(\cos\theta-\e)}{\cD} & \e\sin\theta\\[8pt]
\dfrac{\e(\cos\theta+\e)}{\cD} & \dfrac{\e\sin\theta}{\cD} & 1-\e\cos\theta\\[8pt]
-\nu^2\,\dfrac{\e\sin\theta}{\cD} & -\nu^2\,\dfrac{1-\e\cos\theta}{\cD} & 0
\end{pmatrix}.
\end{equation}
Besides $\bq$ in \eqref{eq:Gq}, it is convenient to introduce the $2\pi$-periodic vector
\begin{equation}\label{eq:pq}
\bp(\theta)=\begin{pmatrix}1-\e\cos\theta\\ -\e\sin\theta\end{pmatrix},\qquad\text{so that}\qquad
\bq=J\bp,\qquad J=\begin{pmatrix}0&-1\\ 1&0\end{pmatrix};
\end{equation}
up to the reflection $R$ of Proposition~\ref{prop:equiv}, $\bp$ and $\bq$ are tangent and normal to the isopycnal through the parcel. They satisfy
\begin{equation}\label{eq:pqprop}
\begin{gathered}
\bp\cdot\bq=0,\qquad |\bp|^2=|\bq|^2=h(\theta):=1-2\e\cos\theta+\e^2,\\
\bp'=-G\bp,\qquad \bq'=G^T\bq,\qquad G-G^T=\frac{2\e^2}{\cD}\,J .
\end{gathered}
\end{equation}

\noindent
Let $\mathbb{F}(\theta)$ denote the fundamental matrix satisfying
\begin{equation}
\mathbb{F}'(\theta)=\mathbb{A}(\theta)\mathbb{F}(\theta),
\qquad
\mathbb{F}(0)=\mathbb{I}.
\end{equation}
Since $\mathbb{A}$ is $2\pi$-periodic, the fundamental matrix satisfies
\begin{equation}
\mathbb{F}(\theta+2\pi)=\mathbb{F}(\theta)\mathbb{M},
\qquad
\mathbb{M}:=\mathbb{F}(2\pi),
\end{equation}
where $\mathbb{M}$ is the monodromy matrix. Its eigenvalues $\mu$ are the Floquet multipliers and determine the exponential stability of the periodic system. In particular, if $\mathbf{E}(0)$ is an eigenvector of $\mathbb{M}$ associated with the multiplier $\mu$, then
\begin{equation}
\mathbf{E}(2\pi n)=\mu^n\mathbf{E}(0),
\qquad n\in\mathbb{N},
\end{equation}
  the corresponding solution grows or decays at the rate $\ln|\mu|/(2\pi)$ per unit $\theta$ and, since $|\dot\theta|=\Omega$, at the rate
  \begin{equation}
  \sigma=\frac{\Omega}{2\pi}\ln|\mu|
  \end{equation}
  in physical time.
Hence, $|\mu|<1$ and $|\mu|>1$ correspond to exponential decay and growth, respectively, while $|\mu|=1$ implies zero exponential growth rate.

\noindent The direction in which $\theta$ is traversed does not affect stability. Along a trajectory, $\theta=\pi-k(a_1+c_*t)$, so that $\dot\theta=-kc_*$ and $|\dot\theta|=\Omega$. For $c_*>0$, a time step $t\mapsto t+T$, with $T:=2\pi/\Omega$, corresponds to $\theta\mapsto\theta-2\pi$, and $\mathbb F(\theta-2\pi)=\mathbb F(\theta)\,\mathbb M^{-1}$ gives
\begin{equation}\label{eq:oneperiod}
\mathbf E(t+T)=\mathbb M_t\,\mathbf E(t),\qquad \mathbb M_t:=\mathbb F\bigl(\theta(t)\bigr)\,\mathbb M^{-1}\,\mathbb F\bigl(\theta(t)\bigr)^{-1},
\end{equation}
so the one-period map in physical time is similar to $\mathbb M^{-1}$. For $c_*<0$, $\theta$ increases with time and the same argument gives a map similar to $\mathbb M$. Since the multipliers have the form $\{1,\mu,\mu^{-1}\}$ (Proposition~\ref{prop:mult} below), $\mathbb M$ and $\mathbb M^{-1}$ have the same multipliers, and the sign of $c_*$ is irrelevant for stability.

\begin{proposition}[Structure of the multipliers; {\citealp[Lemmas~4--5, Prop.~6]{P}}]\label{prop:mult}
\leavevmode
\begin{enumerate}[label=(\roman*),leftmargin=1em]
\item $\operatorname{tr}\mathbb A=0$; hence $\det\mathbb F(\theta)\equiv1$ and $\mu_1\mu_2\mu_3=1$.
\item With $\Sigma:=\mathrm{diag}(1,-1,1)$, one has $\Sigma\,\mathbb A(-\theta)\,\Sigma=-\mathbb A(\theta)$; consequently $\mathbb M=\Sigma\,\mathbb M^{-1}\Sigma$, and the spectrum of $\mathbb M$ is invariant under $\mu\mapsto1/\mu$.
\item One multiplier equals $1$ for all values of the parameters, and the other two are the roots $\mu^{\pm1}$ of
\begin{equation}\label{eq:charpoly}
\mu^2-(\operatorname{tr}\mathbb M-1)\,\mu+1=0 .
\end{equation}

\end{enumerate}
\end{proposition}
\noindent
 We recall (see {\citealp[Remark~7]{P}} for the details), that the multiplier $1$ expresses the conservation of the scalar
\begin{equation}\label{eq:I}
\mathcal I:=\nu^2\,\bp\cdot\tV-2\e^2\,\tM, 
\end{equation}
along every solution of \eqref{eq:sysP}. Up to a factor that is constant along each trajectory, $\mathcal{I}$ is the leading-order WKB amplitude of the Ertel potential vorticity of the perturbation.\\
With the structure presented above, we can state the following stability criterion.
\begin{theorem}[Stability criterion; {\citealp[Cor.~8]{P}}]\label{thm:criterion}
Let
\begin{equation}\label{eq:Delta}
\Delta(\e,\nu^2):=\operatorname{tr}\mathbb M-1=\mu+\mu^{-1}\in\mathbb R .
\end{equation}
The amplitudes grow exponentially in $t$ if and only if $|\Delta|>2$, in which case the maximal growth rate is
\begin{equation}\label{eq:sigmaP}
\sigma=\frac{\Omega}{2\pi}\,\operatorname{arcosh}\frac{|\Delta|}{2}.
\end{equation}
If $|\Delta|<2$ all solutions are bounded; if $|\Delta|=2$ there is no exponential growth, but solutions may grow algebraically. When $\Delta<-2$ the growing multiplier is negative, and the growing solution reverses its sign at every orbital period.
\end{theorem}

\noindent
We will use three results from \citet{P}. The first concerns a neutrally stratified level, namely $\nu=0$, for which the analysis is explicit \citep[\S3.1.5, Thm.~11]{P}. The last equation of \eqref{eq:sysP} reduces to $\tM'=0$, so $\tM$ is constant and the monodromy matrix is block triangular,
\begin{equation}\label{eq:Mblock}
\mathbb M=\begin{pmatrix}X(2\pi)&\ast\\ 0&1\end{pmatrix},\qquad\text{hence}\qquad \Delta=\operatorname{tr}X(2\pi),
\end{equation}
where $X(\theta)$ is the fundamental matrix of $\tV'=G(\theta)\tV$, with $X(0)=\mathbb I$. Up to the factor $kc_*$ and the reflection $R$ of Proposition~\ref{prop:equiv}, $G$ is the velocity gradient of the basic flow along the trajectory, and we split it into its symmetric part, the strain, and its antisymmetric part, which is proportional to the vorticity of the wave:
\begin{equation}\label{eq:Gsplit}
G(\theta)=\frac{\e}{\cD}\begin{pmatrix}-\sin\theta&\cos\theta\\ \cos\theta&\sin\theta\end{pmatrix}+\frac{\e^2}{\cD}\,J,\qquad J=\begin{pmatrix}0&-1\\ 1&0\end{pmatrix}.
\end{equation}
The strain matrix has eigenvalues $\pm1$. To find its principal axes, let $\mathcal R_\chi$ be the rotation by the angle $\chi$; then
\begin{equation}\label{eq:Rchi}
\mathcal R_\chi=\begin{pmatrix}\cos\chi&-\sin\chi\\ \sin\chi&\cos\chi\end{pmatrix},\qquad
\mathcal R_\chi\begin{pmatrix}1&0\\ 0&-1\end{pmatrix}\mathcal R_\chi^T=\begin{pmatrix}\cos2\chi&\sin2\chi\\ \sin2\chi&-\cos2\chi\end{pmatrix},
\end{equation}
which coincides with the matrix in \eqref{eq:Gsplit} if $\cos2\chi=-\sin\theta$ and $\sin2\chi=\cos\theta$, that is, if $2\chi=\theta+\pi/2$. Hence
\begin{equation}\label{eq:Gstrain}
G(\theta)=\frac{\e}{\cD}\,\mathcal R_{\chi(\theta)}\begin{pmatrix}1&0\\ 0&-1\end{pmatrix}\mathcal R_{\chi(\theta)}^T+\frac{\e^2}{\cD}\,J,\qquad \chi(\theta)=\frac{\theta}{2}+\frac{\pi}{4}:
\end{equation}
the principal axes of the strain, which are the columns of $\mathcal R_{\chi(\theta)}$, rotate at half the orbital rate. We therefore write the amplitude in the frame that rotates with them,
\begin{equation*}
\tV(\theta)=\mathcal R_{\chi(\theta)}\,\mathbf W(\theta).
\end{equation*}
Since $\mathrm d\mathcal R_\chi/\mathrm d\chi=\mathcal R_\chi J$ and $\chi'=1/2$, one has $\tV'=\mathcal R_\chi\bigl(\mathbf W'+\tfrac12J\mathbf W\bigr)$, and since rotations commute with $J$, $\mathcal R_\chi^TJ\mathcal R_\chi=J$. Inserting \eqref{eq:Gstrain} into $\tV'=G\tV$ and multiplying by $\mathcal R_\chi^T$ gives the autonomous system
\begin{equation}\label{eq:K}
\begin{gathered}
\mathbf W'=K\mathbf W,\qquad
K=\frac{\e}{\cD}\begin{pmatrix}1&0\\ 0&-1\end{pmatrix}+\Bigl(\frac{\e^2}{\cD}-\frac12\Bigr)J=\begin{pmatrix}\eta&-\varpi\\ \varpi&-\eta\end{pmatrix},\\
\eta:=\frac{\e}{\cD},\qquad \varpi:=\frac{\e^2}{\cD}-\frac12=\frac{3\e^2-1}{2\cD}.
\end{gathered}
\end{equation}
In this frame the perturbation sees a steady strain of rate $\eta$ and a rotation of rate $\varpi$, the rotation due to the vorticity of the wave minus that of the frame. Since $\operatorname{tr}K=0$ and $\det K=\varpi^2-\eta^2$, the eigenvalues of $K$ are $\pm\Lambda$, with
\begin{equation}\label{eq:Lambda}
\Lambda^2=\eta^2-\varpi^2=\frac{4\e^2-(3\e^2-1)^2}{4\cD^2}=\frac{(9\e^2-1)(1-\e^2)}{4\cD^2}=\frac{9\e^2-1}{4(1-\e^2)} .
\end{equation}
The fundamental matrix is $X(\theta)=\mathcal R_{\chi(\theta)}\,e^{K\theta}\,\mathcal R_{\chi(0)}^T$. During one period the strain axes turn by $\pi$, $\chi(2\pi)=\chi(0)+\pi$, so that $\mathcal R_{\chi(2\pi)}=-\mathcal R_{\chi(0)}$ and
\begin{equation}\label{eq:Delta0}
X(2\pi)=-\mathcal R_{\chi(0)}\,e^{2\pi K}\,\mathcal R_{\chi(0)}^T,\qquad
\Delta(\e,0)=-\operatorname{tr}e^{2\pi K}=-2\cosh(2\pi\Lambda).
\end{equation}
Three cases follow. If $\e>1/3$, then $\Lambda>0$, the strain dominates the rotation and $\Delta<-2$: the level is unstable, with growth rate $\sigma=\Omega\Lambda$, and the growing multiplier $-e^{2\pi\Lambda}$ is negative, so that the growing mode changes sign at every orbital period, during which the strain axes turn by $\pi$. An exact growing solution is $\tM\equiv0$, $\tV(\theta)=e^{-\Lambda\theta}\,\mathcal R_{\chi(\theta)}\,(\varpi,\ \eta+\Lambda)^T$, built on the eigenvector of $K$ for the eigenvalue $-\Lambda$: it grows as $\theta$ decreases, that is, in physical time when $c_*>0$, while for $c_*<0$ the growing solution is $e^{\Lambda\theta}\,\mathcal R_{\chi(\theta)}\,(\varpi,\ \eta-\Lambda)^T$. If $\e=1/3$, then $\eta=-\varpi=3/8$, $K^2=0$ and $e^{K\theta}=\mathbb I+K\theta$, so that the growth is algebraic (linear in $\theta$). If $\e<1/3$, then $\Lambda=\mathrm i|\Lambda|$ with $0<|\Lambda|<1/2$, $\Delta(\e,0)=-2\cos(2\pi|\Lambda|)\in(-2,2)$, and every solution is bounded. For a homogeneous incompressible liquid, $\nu=0$ and the particle paths \eqref{eq:paths} with $\alpha=u=v=0$ are those of Gerstner's wave, so that \eqref{eq:Delta0} recovers the threshold $\e>1/3$ found by \citet{L04}, who developed the short-wavelength method for flows given in Lagrangian form and applied it to Gerstner's wave; the computation above makes the mechanism explicit, since instability occurs exactly when the strain exceeds the effective rotation felt by the perturbation in the frame of the principal axes of strain, $\eta>|\varpi|$.\\Second, $\Delta(\e,\nu^2)$ is real-analytic on $(0,1)\times\mathbb R$, and at $(1/3,0)$ \citep[Thm.~13]{P}
\begin{equation}\label{eq:derivs}
\frac{\partial\Delta}{\partial\e}=-\frac{27\pi^2}{4},\qquad \frac{\partial\Delta}{\partial(\nu^2)}=-3\pi^2,
\end{equation}
so that, near $(1/3,0)$, the unstable set $\mathcal S:=\{|\Delta|>2\}$ is $\{\e>\e_c(\nu^2)\}$, with $\e_c$ real-analytic and
\begin{equation}\label{eq:ecP}
\e_c(\nu^2)=\frac13-\frac49\,\nu^2+O(\nu^4).
\end{equation}
Third, since $\mu\mu^{-1}=1$, a real pair of multipliers cannot change sign without leaving the real axis, so for $\nu^2<0$ the convective instability ($\Delta>2$) and the wave-induced one ($\Delta<-2$) are separated by a stable band $\e_a<\e<\e_b$, with $\e_a\approx(|\nu^2|/4)^{1/4}$ and $\e_b\approx\frac13+\frac49|\nu^2|$ for small $|\nu^2|$ \citep[\S3.1.6]{P}.
 

\subsection{Instability results for Godin's solution}\label{sec:results}
We now apply these results to Godin's waves. By Proposition~\ref{prop:equiv}, every level of a Godin wave corresponds to a point $(\e,\nu^2)$ of the parameter plane of \citet{P}, with $\nu^2=(1-\e^2)\nu_0^2$, so that the stability of the whole wave field is decided by the discriminant $\Delta$. Three features are new with respect to \citet{P}. First, the stratification is measured by the buoyancy frequency $N_0$ of the hydrostatic reference state, through $\nu_0^2=\cos^2\!\alpha\,N_0^2/\Omega^2$: the natural stability diagram is the $(\e,\nu_0^2)$-plane, in which an isothermal layer lies on a single horizontal line. Second, the results hold for Godin's whole class: surface and edge waves, for which only the effective gravity $g\cos\alpha$ acts in the plane of motion, waves on uniform currents, and a general equation of state. Third, the consequences are different: in an isentropic fluid the threshold $\e>1/3$ of \citet{L04} extends to compressible surface and edge waves (\S~\ref{cor:isentropic}), weak stratification lowers it by $\frac{32}{81}\nu_0^2$ (\S~\ref{cor:weak}), and the resonance tongues issuing from $\nu_0=n/2$, which were out of range for the mountain waves of \citet{P}, become relevant for the fast acoustic-gravity waves of \citet{G} (\S~\ref{cor:poly} and Section~\ref{sec:physics}).\\
By Proposition~\ref{prop:equiv}, the amplitudes of a perturbation with $\bxi=\be_2$ on the level $(a_2,a_3)$ obey the system \eqref{eq:sysP} studied by \citet{P}, with $\e=e^{ka_2}$ and $\nu^2=(1-\e^2)\,\nu_0^2(b)$. The results recalled in Section~\ref{sec:reduction} therefore hold level by level, and it remains to express them in terms of the steepness $\e$ and of the reference stratification $\nu_0^2$. One Floquet multiplier equals $1$ and the other two are $\mu$ and $\mu^{-1}$ (Proposition~\ref{prop:mult}); the multiplier $1$ corresponds to the conservation of $\mathcal I$ in \eqref{eq:I}. The identification of $\mathcal I$ with the leading-order Ertel potential vorticity of the perturbation \citep[App.~A]{P} carries over to Godin's fluids, with the specific entropy $S$ in place of $\ln\Theta$: $S$ is materially conserved, and $(\nabla\rho\times\nabla P)\cdot\nabla S=0$ because $\rho=\rho(P,S)$. 
Finally, by Theorem~\ref{thm:criterion}, the level is exponentially unstable if and only if $|\Delta(\e,\nu^2)|>2$, with the growth rate \eqref{eq:sigmaP}, in which $\Omega=\sqrt{gk\cos\alpha}$ is the orbital frequency of the particles.\\
\\
Since $\nu^2=(1-\e^2)\,\nu_0^2$, a level is unstable if and only if the point $(\e,(1-\e^2)\nu_0^2)$ lies in the unstable region of \citet[figure~5]{P}; in the $(\e,\nu_0^2)$-plane of figure~\ref{fig:diagram}, the stability diagram is therefore that of \citet{P}, with the stratification axis stretched, at each steepness, by the factor $1/(1-\e^2)$. The limit of small steepness can be analysed explicitly. The coefficients of \eqref{eq:sysP} depend analytically on $(\e,\nu^2)$ for $0\le\e<1$, so that $\Delta$ is continuous up to $\e=0$ \citep[Thm.~13]{P} and $\Delta\bigl(\e,(1-\e^2)\nu_0^2\bigr)\to\Delta(0,\nu_0^2)$ as $\e\to0$. At $\e=0$ the coefficients no longer depend on $\theta$: $G\equiv0$ and $\bq\equiv(0,1)^T$ in \eqref{eq:Gq}, so that $\widetilde V_1$ is constant and $\widetilde V_2''=-\nu_0^2\,\widetilde V_2$. The perturbation is then a buoyancy oscillation of frequency $N_0\cos\alpha=\Omega\nu_0$, whose discriminant over one wave period is $\Delta(0,\nu_0^2)=2\cos(2\pi\nu_0)$ if $\nu_0^2\ge0$ and $\Delta(0,\nu_0^2)=2\cosh(2\pi|\nu_0|)$ if $\nu_0^2<0$.

In a stably stratified fluid, therefore, $|\Delta(0,\nu_0^2)|\le2$, with equality only at $\nu_0=n/2$, $n=1,2,\dots$, where the buoyancy frequency $N_0\cos\alpha$ is an integer multiple of half the orbital frequency $\Omega$, the classical condition for parametric resonance. By continuity, waves of small steepness can destabilise a stably stratified level only near these resonances, from which the unstable tongues of figure~\ref{fig:diagram} emanate; since $\Delta(0,\nu_0^2)=2(-1)^n$ there, the instability is subharmonic ($\Delta<-2$) for odd $n$ and harmonic ($\Delta>2$) for even $n$. In a statically unstable reference state, instead, $\Delta(0,\nu_0^2)>2$, so that a level with $\nu_0^2<0$ is unstable whenever its steepness is small enough, with a growth rate \eqref{eq:sigmaP} that tends to the convective value $\Omega|\nu_0|=|N_0|\cos\alpha$ as $\e\to0$. At larger steepness this convective instability is separated from the wave-induced one by the thin stable band recalled in Section~\ref{sec:reduction}; at $\nu_0^2=-0.1$, for instance, the levels with $0.334<\e<0.366$ are stable. These results are summarised in the following theorem.

\begin{theorem}\label{thm:main}
Let $u,v$ be constant and fix a level $(a_2,a_3)$ of the fluid, with $\e=e^{ka_2}$ and $\nu^2=(1-\e^2)\nu_0^2(b)$. For the perturbations \eqref{eq:xi_e2} with $\boldsymbol{\xi}(0)=\be_2$:
\begin{enumerate}[label=(\roman*),leftmargin=1em]
\item the Floquet multipliers are $\{1,\mu,\mu^{-1}\}$; the multiplier $1$ is carried by the conserved quantity $\mathcal I=\nu^2\,\bp\cdot\tV-2\e^2\tM$, proportional to the leading-order Ertel potential vorticity of the perturbation (Proposition~\ref{prop:mult} and \eqref{eq:I});
\item the level is exponentially unstable if $|\Delta(\e,\nu^2)|>2$, where $\Delta$ is the discriminant \eqref{eq:Delta}, and the growth rate is
\begin{equation}\label{eq:sigma}
\sigma=\frac{\Omega}{2\pi}\,\operatorname{arcosh}\frac{|\Delta|}{2},\qquad \Omega=\sqrt{gk\cos\alpha};
\end{equation}
\item as $\e\to0$, $\Delta\to2\cos(2\pi\nu_0)$ if $\nu_0^2\ge0$ and $\Delta\to2\cosh(2\pi|\nu_0|)$ if
$\nu_0^2<0$; hence waves of small steepness can destabilise a stably stratified level only near the
resonances $\nu_0=n/2$, $n=1,2,\dots$, whereas a level with $\nu_0^2<0$ is unstable whenever its
steepness is small enough, with a growth rate that tends to $|N_0|\cos\alpha$.
\end{enumerate}
\end{theorem}

\subsubsection{Isentropic fluids}\label{cor:isentropic}
If $N_0\equiv0$, it follows from \eqref{eq:change} that $\nu^2=\nu_0^2=0$, and a level is unstable if and only if $\e>1/3$. The instability is a subharmonic parametric resonance, $\Delta=-2\cosh(2\pi\Lambda)<-2$ with $\Lambda=\frac12\sqrt{(9\e^2-1)/(1-\e^2)}$, and its growth rate is
\begin{equation}\label{eq:sigma0}
\sigma=\Omega\Lambda=\frac{\Omega}{2}\sqrt{\frac{9\e^2-1}{1-\e^2}},\qquad \Omega=\sqrt{gk\cos\alpha}.
\end{equation}
This covers the homogeneous incompressible liquid, for which the case $\alpha=0$ is Leblanc's result \citet{L04}; any isentropic or barotropic fluid \cite[(4.6)]{G}; and the adiabatic polytropic atmosphere $n=\gamma$ \cite[(4.9)]{G}. \\
For $\alpha>0$ it gives the instability of compressible Gerstner edge waves with the same threshold, extending the incompressible edge-wave analysis of \citet{IK14}.

\subsubsection{Weak reference stratification}\label{cor:weak}
Near the neutral threshold $\e=1/3$ given by \eqref{eq:Delta0}, the effect of a weak reference stratification can be computed explicitly. Let $\hat\Delta(\e,\nu_0^2):=\Delta\bigl(\e,(1-\e^2)\nu_0^2\bigr)$ be the discriminant written in terms of the reference stratification; like $\Delta$, it is real-analytic. At $(\e,\nu_0^2)=(1/3,0)$, where $\hat\Delta=-2$, its derivatives follow from \eqref{eq:derivs} by the chain rule. The derivative with respect to the steepness is unchanged, $\partial\hat\Delta/\partial\e=\partial\Delta/\partial\e=-27\pi^2/4$, because the term that comes from the factor $1-\e^2$ is multiplied by $\nu_0^2=0$, while the derivative with respect to the stratification acquires the factor $1-\e^2=8/9$, so that $\partial\hat\Delta/\partial(\nu_0^2)=\frac89\,(-3\pi^2)=-8\pi^2/3$. Since $\partial\hat\Delta/\partial\e\neq0$, the implicit function theorem applied to $\hat\Delta=-2$ gives, near $(1/3,0)$, a real-analytic curve $\e=\ech(\nu_0^2)$ with $\ech(0)=1/3$ and slope $\ech'(0)=-(-8\pi^2/3)/(-27\pi^2/4)=-32/81$, that is,
\begin{equation}\label{eq:ec}
\ech(\nu_0^2)=\frac13-\frac{32}{81}\,\nu_0^2+O(\nu_0^4).
\end{equation}
Near $(1/3,0)$, moreover, $\hat\Delta$ is close to $-2$ and decreases as $\e$ increases, so that $|\hat\Delta|>2$ if and only if $\hat\Delta<-2$, that is, if and only if $\e>\ech(\nu_0^2)$. For weak stratification and steepness close to $1/3$, a level is therefore unstable if and only if its steepness exceeds $\ech(\nu_0^2)$: a stable stratification lowers the critical steepness below $1/3$, and an unstable one raises it, $\ech$ being then the upper edge of the stable band of Section~\ref{sec:reduction}. The coefficient $-32/81$ is the coefficient $-4/9$ of \eqref{eq:ecP} multiplied by $1-\e^2=8/9$, because on the neutral threshold $\nu^2=\frac89\,\nu_0^2$.


\subsubsection{Polytropic and isothermal atmospheres}\label{cor:poly}
For an ideal gas with constant $\gamma$ and a polytropic reference state $P_s\propto\rho_s^{\,n}$, inserting the buoyancy frequency \eqref{eq:N0poly} and $\Omega^2=gk\cos\alpha$ into $\nu_0^2=\cos^2\!\alpha\,N_0^2/\Omega^2$ gives
\begin{equation}\label{eq:nu0poly}
\nu_0^2(b)=\frac{\cos\alpha}{k\,H(b)}\Bigl(\frac1n-\frac1\gamma\Bigr),\qquad H=\frac{R_dT}{g}:
\end{equation}
the stratification parameter is proportional to the wavelength measured in scale heights, and it depends on the level only through the temperature $T(b)$. In an isothermal layer ($n=1$), $1-1/\gamma=\kappa$ and $\nu_0^2=\kappa\cos\alpha/(kH)$ is the same at every level: the whole wave field lies on one horizontal line of the $(\e,\nu_0^2)$-diagram of figure~\ref{fig:diagram}, and its stability is read off that line. The principal subharmonic resonance $\nu_0=1/2$ occurs at $kH=4\kappa\cos\alpha$, that is, at the wavelength $\lambda=\pi H/(2\kappa\cos\alpha)=\frac{7\pi}{4}H/\cos\alpha\approx5.5\,H/\cos\alpha$ for $\kappa=2/7$, and the harmonic resonance $\nu_0=1$ at $\lambda=2\pi H/(\kappa\cos\alpha)=7\pi H/\cos\alpha\approx22\,H/\cos\alpha$. In a super-adiabatic layer ($n>\gamma$), instead, $\nu_0^2<0$, and by Theorem~\ref{thm:main}(iii) every level is unstable whenever its steepness is small enough; at larger steepness the convective and the wave-induced instabilities are separated by the thin stable band of Section~\ref{sec:reduction}.

\subsection{The stability diagram}\label{sec:numerics}
Figure~\ref{fig:diagram}(\textit{a}) shows the growth rate of the instability in the $(\e,\nu_0^2)$-plane; the numerical method is described in Appendix~\ref{app:numerics}.

\begin{figure}
\centering
\includegraphics[width=\textwidth]{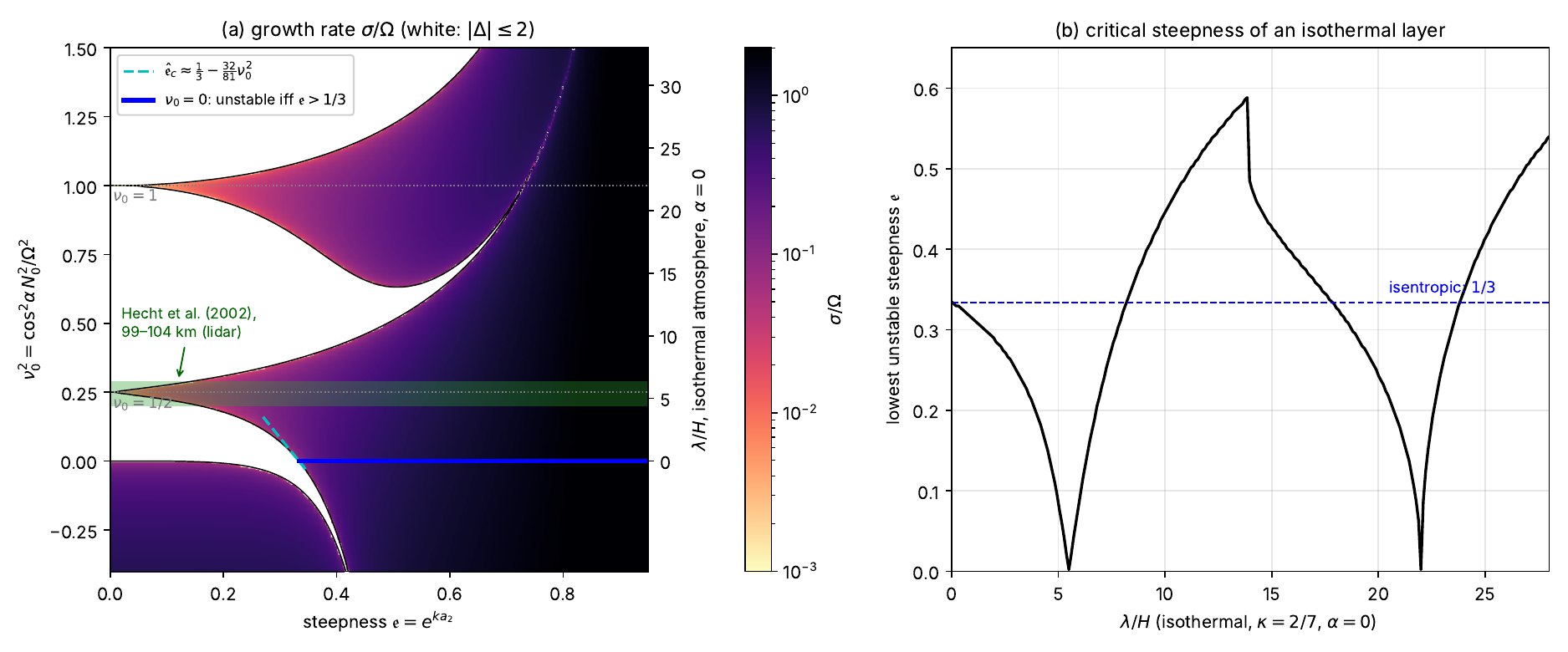}
\caption{(\textit{a}) Growth rate $\sigma/\Omega$ in the $(\e,\nu_0^2)$-plane, white where $|\Delta|\le2$, with the curves $|\Delta|=2$ in black. The right axis converts $\nu_0^2$ into $\lambda/H$ for an isothermal atmosphere ($\alpha=0$, $\kappa=2/7$). The green band is the range of $\nu_0^2$ measured between 99 and 104\,km for the wave observed by \citet{H02} (Section~\ref{sec:physics}). (\textit{b}) Lowest unstable steepness along each horizontal line of (\textit{a}), as a function of $\lambda/H$ for an isothermal layer.}
\label{fig:diagram}
\end{figure}

The neutral curve leaves $(1/3,0)$ with the slope $-32/81$ of \eqref{eq:ec} and ends at the tip $(0,1/4)$ of the tongue issuing from $\nu_0=1/2$; including the next term, the computed curve is $\ech\approx\frac13-\frac{32}{81}\nu_0^2-0.922\,\nu_0^4$, with a relative error below $1.5\,\%$ for $\nu_0^2\le0.1$. For $0.20\le\nu_0^2\le0.30$ the critical steepness is below $0.18$, and it vanishes at $\nu_0^2=1/4$. In an isothermal layer (panel \textit{b}) the critical steepness is $1/3$ for short waves, falls to zero at $\lambda/H=7\pi/4\approx5.50$, where $\nu_0=1/2$, is back at $1/3$ at $\lambda/H\approx8.2$ and peaks at $0.587$ at $\lambda/H\approx13.9$ ($\nu_0^2\approx0.631$). For longer waves the lowest unstable steepness lies in the tongue issuing from $\nu_0=1$, where it is below $1/3$ for $17.8\lesssim\lambda/H\lesssim23.8$ and vanishes at $\lambda/H=7\pi\approx22.0$. For $\nu_0^2>0.631$ a horizontal line crosses both this tongue, where $\Delta>2$, and the main unstable region, where $\Delta<-2$; since $\Delta$ must cross $[-2,2]$ to change sign, the two are separated by a stable gap, which narrows as $\nu_0^2$ grows, from $0.517<\e<0.587$ at $\nu_0^2=0.632$ to $0.661<\e<0.669$ at $\nu_0^2=0.8$.

\subsection{Where the instability lives}\label{sec:where}

\subsubsection{Surface waves ($\alpha=0$)} The steepness $\e=e^{ka_2}$ increases with $a_2$ up to its surface value $\e_0:=e^{ka_2^{(0)}}$. The free surface is a trochoid of crest-to-trough height $H_w=2\e_0/k$, so $\e_0=\pi H_w/\lambda$. In an isentropic fluid the wave is therefore unstable if and only if
\begin{equation}\label{eq:surf}
a_2^{(0)}>-\frac{\ln3}{k}\qquad\Longleftrightarrow\qquad \frac{H_w}{\lambda}>\frac{1}{3\pi}\approx0.106,
\end{equation}
\begin{figure}
\centering
\includegraphics[width=.8\textwidth]{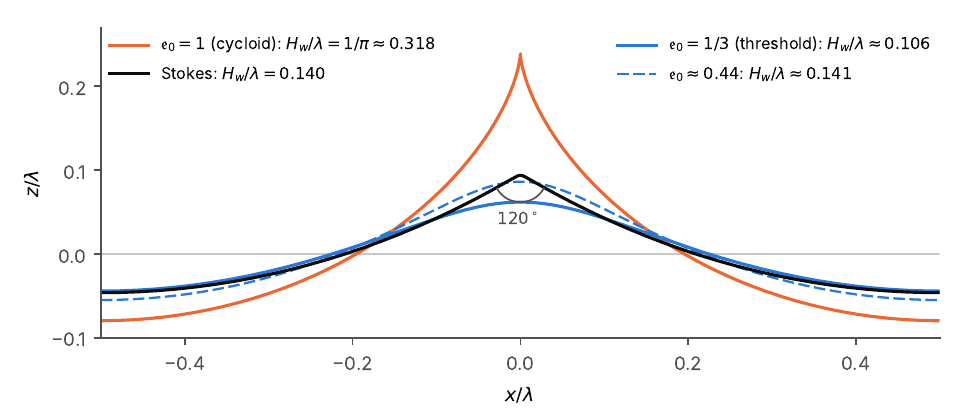}
\caption{Free surfaces of waves of wavelength $\lambda$, with zero mean level. Blue and orange: Godin's waves, which have the kinematics of Gerstner's, with surface steepness $\e_0=1/3$, the instability threshold ($H_w/\lambda=1/(3\pi)$); $\e_0\approx0.44$, with the height of the highest Stokes wave (dashed); and $\e_0=1$, the steepest Gerstner wave, a cycloid with cusped crests ($H_w/\lambda=1/\pi$). Black: a Stokes wave with $H_w/\lambda=0.140$, close to the highest one ($H_w/\lambda\approx0.141$, crest angle $120^\circ$), computed from the equation of \citet{Babenko87}.}
\label{fig:surfaces}
\end{figure}
to be compared with $H_w/\lambda=1/\pi\approx0.318$ for the steepest Gerstner wave, $\e_0=1$, whose free surface is a cycloid with cusped crests: the instability sets in at one third of the maximal height. The threshold also lies below the height of the highest irrotational (Stokes) wave, $H_w/\lambda\approx0.141$, whose crest encloses an angle of $120^\circ$ \citep{Stokes1880,Williams1981}, so that a Gerstner wave as high as the highest Stokes wave is already unstable (figure~\ref{fig:surfaces}). The unstable levels form the layer $-\ln3/k<a_2<a_2^{(0)}$, of mean thickness $k^{-1}\ln(3\e_0)$. By \S~\ref{cor:weak}, a weak and uniform stable reference stratification moves the base of this layer down by $\frac{32}{27}\nu_0^2/k+O(\nu_0^4)$: differentiate $k^{-1}\ln\ech$ and use $\ech(0)=1/3$.

\subsubsection{Edge waves ($\alpha>0$)} Write $\beta(a_2):=a_2-e^{2ka_2}/(2k)$, which is increasing, so that $b=\beta(a_2)\cos\alpha-a_3\sin\alpha$. The fluid occupies $\{a_3\le0,\ b\le b_0\}$, and at distance $x_2=-a_3$ from the wall the free surface sits at $\beta(a_2)=(b_0-x_2\sin\alpha)/\cos\alpha$. The steepest particles are therefore on the wall at the shoreline, where $\e_*:=e^{ka_2^*}$ with $\beta(a_2^*)=b_0/\cos\alpha$; admissibility requires $a_2^*<0$. Up to the orbital excursion $O(\e/k)$, the label $a_2$ is the up-slope coordinate $x_3$. For a uniform $\nu_0^2$ with critical steepness $\e_c$, the unstable set $\{a_2>k^{-1}\ln\e_c\}$ is a wedge bounded by the wall, the free surface and the plane $a_2=k^{-1}\ln\e_c$:
\begin{itemize}[leftmargin=2em]
\item it extends a distance
$L:=\beta(a_2^*)-\beta(k^{-1}\ln\e_c)=k^{-1}\bigl[\ln(\e_*/\e_c)-(\e_*^2-\e_c^2)/2\bigr]$ down the slope
from the mean shoreline;
\item its thickness normal to the wall is at most $L\cot\alpha$;
\item it is non-empty if and only if $\e_*>\e_c$, that is $\e_*>1/3$ when $N_0\equiv0$.
\end{itemize}
\begin{figure}
\centering
\includegraphics[width=.8\textwidth]{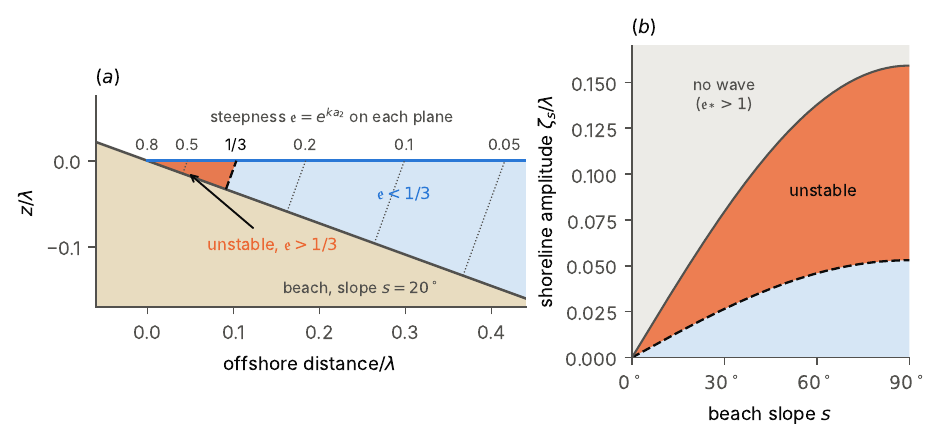}
\caption{Instability of Godin's edge waves for $N_0\equiv0$, where the critical steepness is $1/3$. (\textit{a}) Section normal to the shoreline, averaged along the direction of propagation, for a beach slope $s=20^\circ$ and shoreline steepness $\e_*=0.8$: the steepness $\e=e^{ka_2}$ is uniform on planes normal to the beach (dotted, labelled by $\e$), the mean free surface is horizontal, and the unstable region (orange) is the triangle between the beach, the free surface and the plane $\e=1/3$ (dashed). (\textit{b}) Unstable edge waves (orange) in terms of the beach slope $s$ and of the vertical amplitude $\zeta_s=\e_*\sin s/k$ of the shoreline oscillation: the wave is unstable for $\zeta_s/\lambda>\sin s/(6\pi)$ and exists for $\zeta_s/\lambda<\sin s/(2\pi)$, i.e.\ $\e_*<1$.}
\label{fig:edge}
\end{figure}
In terms of the beach slope $s=\pi/2-\alpha$ and the vertical amplitude $\zeta_s=\e_*\sin s/k$ of the shoreline oscillation, an edge wave is unstable when $\zeta_s/\lambda>\sin s/(6\pi)$ for $N_0\equiv0$ (figure~\ref{fig:edge}): on a beach of slope $s=5^\circ$, an edge wave of wavelength 50\,m is unstable as soon as the shoreline oscillates by more than 23\,cm.\\
The factor $\cos^2\!\alpha$ in $\nu_0^2$ makes stratification almost irrelevant on gentle beaches. One has $\nu_0^2=\sin s\;N_0^2/(gk)$; for $s=5^\circ$, $\lambda=50$\,m and $N_0=0.02\,\mathrm{s^{-1}}$ this is $2.8\times10^{-5}$, and $\e_c$ differs from $1/3$ by $1.1\times10^{-5}$. On natural beaches, edge waves are also modified by the bottom topography \citep{KDL81}, which Godin's solution does not include.

\section{Waves on currents}\label{sec:currents}

For a uniform current, with both $u$ and $v$ constant, the reduction of
\S\ref{sec:reduction} remains unchanged. Indeed,
\[
\Omega=k|c+v|=(gk\cos\alpha)^{1/2}
\]
is independent of $v$, which therefore only Doppler-shifts the frequency
observed in the laboratory frame. Theorem~\ref{thm:main} and the results of \S \ref{sec:results} apply verbatim. This is consistent with
\citet[Remark~1]{P}, where the contributions involving $\bar u$ and
$\bar w$ drop out, and includes the configuration used by \citet{G} to
interpret the observations of \citet{H02}; see Section~\ref{sec:physics}.

The same argument applies locally to the discontinuous-current solution
of \citet[(4.12)]{G}. In each layer $m=1,2$, the flow is a
uniform-current wave with the same wavenumber $k$, intrinsic frequency
$\Omega=(gk)^{1/2}$, and wave-amplitude factor $e^{ka_2}$. These
quantities are continuous across the interface $a_2=A$, as is the
pressure, although $\rho$ and $C$ may have finite jumps there
\citep{G}. Since the instability criterion is local along particle
trajectories, and no trajectory crosses $a_2=A$,
Theorem~\ref{thm:main} applies at every point away from the interface,
with the reference stratification chosen from the layer containing that
point. This conclusion also covers Godin's standing-wave configuration,
for which $v_1+v_2=0$. At the interface itself, however, the jump in
tangential velocity produces a vortex sheet and hence the possibility
of Kelvin--Helmholtz instability. Such an interfacial mechanism is not
described by the present short-wavelength analysis.

A qualitatively different situation arises when the transverse current
is sheared. For $\alpha=0$, Godin's solution \cite[(4.10)]{G} permits
$u=u(a_1,a_2)$. Although the velocity field remains independent of
$x_2$, one now has
\(
(\nabla\bU)^T\be_2=\nabla u\neq0,
\)
so that $\be_2$ no longer defines an invariant wavevector direction.
Because the phase $\Phi$ is advected and may therefore be expressed as a
function of the Lagrangian labels, the initial choice $\Phi_0=x_2$
yields
\[
\Phi=-a_3=x_2-ut.
\]
Consequently, the wavevector evolves according to
\begin{equation}\label{eq:xi_shear}
  \bxi(t)
  =\be_2-t\,\nabla u(a_1,a_2),
  \qquad
  \nabla u
  =\partial_{a_1}u\,\nabla a_1
   +\partial_{a_2}u\,\nabla a_2 .
\end{equation}
Thus, unless $\nabla u=0$, the wavevector acquires an in-plane component
that grows linearly in time, progressively tilting it away from the
transverse direction. This behaviour is particularly transparent when $u=u(a_2)$. In that
case,
\[
\nabla u=u'(a_2)\nabla a_2,
\]
and $\bxi$ becomes asymptotically aligned with the isopycnal normal
$\bn(\varphi)$. Since the direction of $\bn(\varphi)$ oscillates along
each orbit, the coefficients of the amplitude equations are no longer
strictly periodic, but only asymptotically so. In particular, the
projector
\(
(\bxi\bxi^T)/|\bxi|^2
\)
evolves from the transverse projector $\be_2\be_2^T$ towards an
oscillating in-plane projector. This is the analogue, in the present
setting, of the secular tilting of wavevectors familiar from sheared
flows. The resulting non-periodic amplitude problem falls outside the
Floquet framework used here and requires a separate asymptotic
analysis; we therefore leave it open. For related results on linear
surface waves propagating over vertically sheared currents, see, for
example, \citet{KC89}.

The potential-vorticity structure also distinguishes the possible
forms of transverse shear. If $u=u(a_2)$, the transverse-shear
vorticity is tangent to the isopycnals, and the Ertel potential
vorticity of the basic state remains zero. By contrast, a dependence of
$u$ on $a_1$ generally introduces a vorticity component across the
isopycnals, so that the basic Ertel potential vorticity no longer
vanishes.
\section{Physical estimates}\label{sec:physics}
\subsection{Ocean}
For waves in the ocean, typically $\lambda=100$\,m, so that $\Omega=(gk)^{1/2}\approx0.79\,\mathrm{s^{-1}}$. Even a strong pycnocline, $N_0=10^{-2}\,\mathrm{s^{-1}}$, gives $\nu_0^2\approx1.6\times10^{-4}$, and the compressibility contribution $g^2/C^2\approx4\times10^{-5}\,\mathrm{s^{-2}}$ to $N_0^2$ is smaller still. Compressibility affects long waves, whose phase speed it reduces through the coupling with hydro-acoustic modes \citep{AKB15,AK17,AKK19}, but its contribution to $\nu_0^2$ is negligible here. By \S~\ref{cor:weak} the threshold moves by $6\times10^{-5}$, so Leblanc's criterion $H_w/\lambda>1/(3\pi)$ applies to Godin's waves in the ocean. Individual ocean waves reach $kH/2=\pi H_w/\lambda\approx0.44$--$0.55$ \citep{TBOW10}, so the threshold $\pi H_w/\lambda=1/3$ lies within the observed range of steepness \citep[for the onset of breaking, see][]{DK16,DKBGT20}.

\subsection{Atmosphere: where the resonance sits} Acoustic-gravity waves are ubiquitous near the mesopause, where they are observed in airglow and lidar data \citep{FA03}, and they can also be radiated into the upper atmosphere by ocean waves \citep{GZB15}. In \citet{P} the physically relevant mountain waves had $\nu^2\lesssim0.1$, and the resonance tongues were out of range. Godin's waves are different. By \S~\ref{cor:poly} an isothermal layer has $\nu_0^2=\kappa\lambda\cos\alpha/(2\pi H)$, and for $\alpha=0$ the orbital frequency exceeds the buoyancy frequency exactly when $\nu_0<1$, i.e.\ for $\lambda<2\pi H/\kappa\approx22\,H$. This is the regime of the fast waves considered in \citet{G}, and it contains the principal subharmonic resonance $\nu_0=1/2$, at $\lambda\approx5.5\,H$ (\S~\ref{cor:poly}).\\
\\
\citet[\S5]{G} matched the observed wavelength $\lambda=35\pm2$\,km, phase speed $160\pm10$\,m\,s$^{-1}$ and period $215\pm10$\,s to the dispersion relation $k(c+v)^2=g$ with a wind component of about $-71$\,m\,s$^{-1}$ along the direction of propagation, i.e.\ opposing it, so that the intrinsic phase speed is $(g/k)^{1/2}\approx230$\,m\,s$^{-1}$; indeed $k\,(160+71)^2\approx9.58$\,m\,s$^{-2}$, close to $g\approx9.53$\,m\,s$^{-2}$ at $90$\,km. Then $\Omega\approx0.041\,\mathrm{s^{-1}}$, an intrinsic period of about $152$\,s. An isothermal layer at $180$--$200$\,K, typical of the mesopause \citep{She00}, would give $\nu_0^2\approx0.26$--$0.29$ by \S~\ref{cor:poly}, i.e.\ $\nu_0\approx0.51$--$0.54$, right at the principal resonance. The sodium (Na) lidar profiles of temperature and wind measured by \citet[figure~3]{H02} during the observation allow a sharper statement (figure~\ref{fig:hecht}; the data and their processing are described in Appendix~\ref{app:hecht}).
\begin{figure}
\centering
\includegraphics[width=\textwidth]{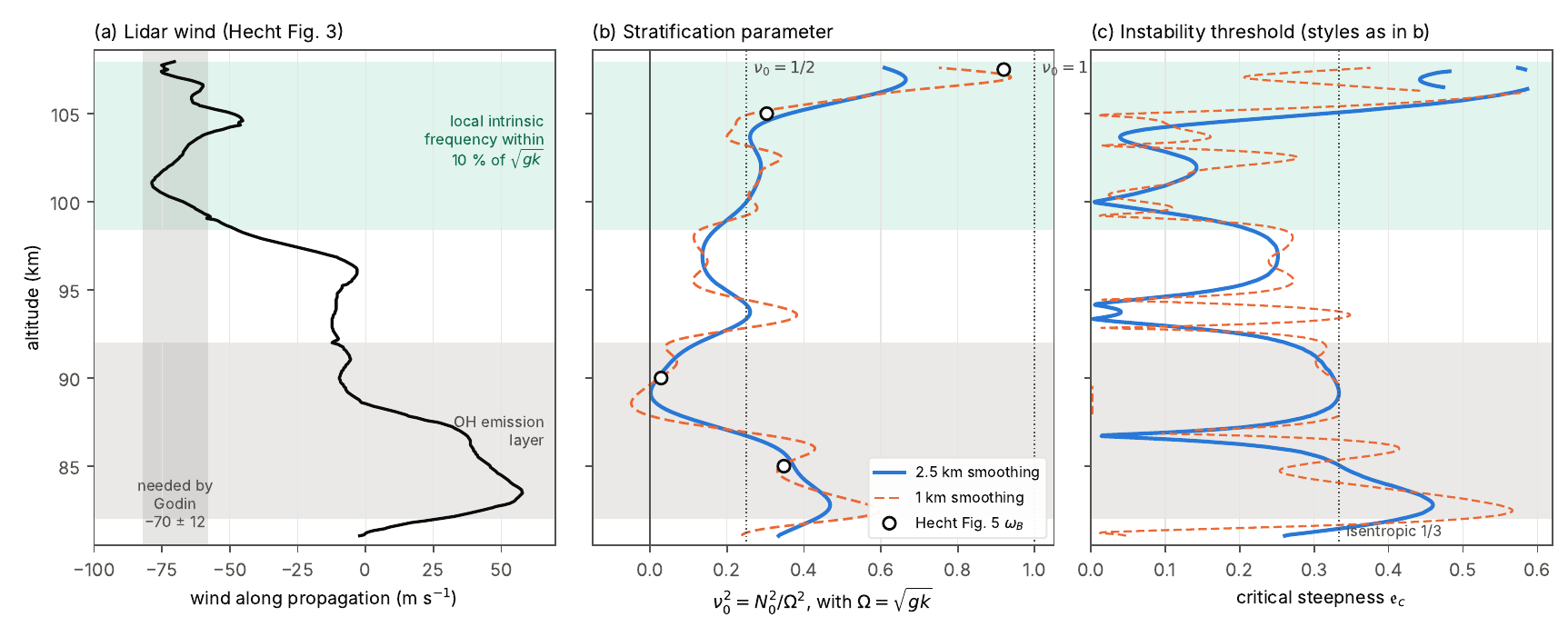}
\caption{The wave of \citet{H02} in Godin's interpretation, from the lidar profiles of their figure~3 (Appendix~\ref{app:hecht}). In all panels the green band marks the altitudes where the local intrinsic frequency $k(c-u)$ is within $10\,\%$ of $(gk)^{1/2}$, and the grey horizontal band is the hydroxyl (OH) emission layer, where the wave was observed. (\textit{a}) Wind along the direction of propagation; the grey vertical band is the wind $-70\pm12$\,m\,s$^{-1}$ required by the dispersion relation \eqref{eq:disp}. (\textit{b}) $\nu_0^2=N_0^2/\Omega^2$, from the temperature smoothed over 2.5\,km (solid) and 1\,km (dashed); circles: buoyancy frequencies of \citet[figure~5]{H02}; dotted lines: $\nu_0=1/2$ and $\nu_0=1$. (\textit{c}) Critical steepness $\e_c$, with the same line styles; dotted line: isentropic value $1/3$.}
\label{fig:hecht}
\end{figure}
They show that Godin's identification is locally consistent only between about 98 and 108\,km. There the wind along the direction of propagation is $-79$ to $-47$\,m\,s$^{-1}$, close to the value $-70\pm12$\,m\,s$^{-1}$ required by \eqref{eq:disp} at 100\,km, and the local intrinsic frequency $k(c-u)$ is within $10\,\%$ of $\Omega$. In the OH emission layer (82--92\,km), where the wave was imaged, the wind is $-12$ to $+58$\,m\,s$^{-1}$ and the intrinsic frequency is only $0.44$--$0.75\,\Omega$, so the motion there is not of Godin's type. Between 99 and 104\,km the measured temperature gives $\nu_0^2=0.20$--$0.29$, with median 0.27 (0.20--0.34 if the profile is smoothed over 1\,km instead of 2.5\,km): the principal subharmonic resonance $\nu_0=1/2$, where the buoyancy period is twice the intrinsic period, lies inside the layer. The critical steepness there has median 0.09 and does not exceed 0.16 (0.12 and 0.28 with the 1\,km smoothing), and it vanishes where $\nu_0^2=1/4$. A steepness of 0.2 throughout the layer would make it unstable everywhere (83\,\% of it with the 1\,km smoothing), with e-folding times of 6--9\,min, shorter than the 18\,min during which the wave was observed; for a steepness of 0.3 the e-folding time drops to about 3\,min. For $\e\lesssim0.05$, on the other hand, the instability is confined to sublayers with $|\nu_0^2-1/4|\lesssim0.013$, and its e-folding time exceeds 30\,min.

The published data do not determine $\e$ near 100\,km. The full-wave model of \citet{H02}, normalised to the observed $2\,\%$ fluctuation of the airglow intensity $I$, gives $\e<0.02$ there (Appendix~\ref{app:hecht}), too small for the instability to matter. If instead the wave were of Godin's type down to the OH layer, every parcel would keep its pressure, density and temperature, so that the intensity would change only because the emitting layer is displaced and alternately stretched and compressed, $I'/\bar I\approx\e$ to first order. A $2\,\%$ fluctuation at 87\,km would then give $\e\approx0.2$ at 100\,km, since $\e=e^{ka_2}$ grows by a factor $e^{k\times13\,\mathrm{km}}\approx10$ between the two altitudes. The unaveraged lidar profiles could settle the question: near 100\,km such a wave would make the horizontal wind oscillate with amplitude $\e c_*\approx45$\,m\,s$^{-1}$ and the sodium layer move up and down by $\e/k\approx1$\,km, with the observed period, while the temperature would change by only 1--4\,K, since each parcel keeps its temperature.

\section{Discussion and conclusions}\label{sec:discussion}
We have shown that the stability of Godin's exact acoustic-gravity waves to short-wavelength perturbations with wavevector normal to the plane of motion reduces, at each level, to a Floquet problem that depends only on the steepness $\e$ and on the stratification parameter $\nu_0^2=\cos^2\!\alpha\,N_0^2/\Omega^2$, where $N_0$ is the buoyancy frequency of the hydrostatic reference state. In isentropic fluids a level is unstable if and only if $\e>1/3$, the threshold found by \citet{L04} for Gerstner's wave, so that compressible surface and edge waves inherit the instability of their incompressible counterparts: a surface wave is unstable for $H_w/\lambda>1/(3\pi)\approx0.106$, a height below that of the highest Stokes wave, and an edge wave as soon as the shoreline oscillates by more than $\lambda\sin s/(6\pi)$ on a beach of slope $s$. Weak stable stratification lowers the critical steepness, which vanishes at the parametric resonances $\nu_0=n/2$; in an isothermal atmosphere the principal one, $\nu_0=1/2$, falls at $\lambda\approx5.5\,H/\cos\alpha$, in the range of the fast acoustic-gravity waves considered by \citet{G}. For the wave observed by \citet{H02}, the lidar temperature places the layer in which Godin's interpretation is consistent on this resonance, where the critical steepness is typically $0.1$ instead of $1/3$ and a steepness of $0.2$ would give e-folding times of 6--9\,min, shorter than the 18\,min of the observation (Section~\ref{sec:physics}). The published data do not determine the steepness of the wave there; the unaveraged lidar profiles could settle the question.

The short-wavelength method is particularly suited to exact solutions given in Lagrangian form \citep{L04,CG,IK14,P26}. Godin's class of incompressible motions of compressible fluids is exactly the class in which its compressible version closes along trajectories, so the Gerstner-type members of the class inherit the instability of \citet{P} level by level: as for the mountain waves of \citet{C23}, the problem reduces to a linear system of ordinary differential equations with periodic coefficients, the amplitude system of \citet{P} with $\nu^2=(1-\e^2)\nu_0^2$, in which uniform currents, the phase speed and the equation of state enter only through $\nu_0^2$. Compared with \citet{P}, three elements are new: the parametrisation by the hydrostatic reference state, natural in Godin's construction, which turns isothermal layers into horizontal lines of the stability diagram; the effective gravity $g\cos\alpha$ of edge waves, whose along-wall component is absorbed by the perturbation pressure and whose inclination reduces the buoyancy by $\cos^2\!\alpha$; and the physical relevance of the resonance tongues, which were out of range for the mountain waves of \citet{P} but contain the fast acoustic-gravity waves, whose $\nu_0$ is of order $1/2$.

The mechanism of the instability is clearest in the form given by \citet[Prop.~10]{P}. The growing and decaying Floquet modes carry no potential vorticity: along a solution with $\mathbf E(2\pi)=\mu\,\mathbf E(0)$ and $\mu\neq1$, the conservation of $\mathcal I$, which is linear in $\mathbf E$, forces $\mathcal I=0$. On these solutions the amplitude $\mathcal A:=\bp\cdot\tV$ obeys the Hill equation
\begin{equation}\label{eq:hill}
\Bigl(\frac{\mathcal A'}{h}\Bigr)'+Q(\theta)\,\mathcal A=0,\qquad
Q(\theta)=\frac{1}{\cD}\left(\nu^2-\frac{4\e^3(\e-\cos\theta)}{h(\theta)^2}\right),
\end{equation}
where $h(\theta):=|\bp(\theta)|^2=1-2\e\cos\theta+\e^2$; its discriminant is $\Delta$. The stiffness $Q$ combines the buoyancy, $\nu^2/\cD$, with a contribution of the wave, which is modulated at the orbital frequency and whose mean is restoring:
\begin{equation}\label{eq:Qmean}
\bar Q:=\frac{1}{2\pi}\int_0^{2\pi}Q\,\mathrm d\theta=\frac{1}{\cD}\left(\nu^2+\frac{4\e^4}{\cD^2}\right).
\end{equation}
If $Q\le0$ along the whole orbit, which happens if and only if $\nu^2\le-4\e^3/(1-\e)^3$, the instability is static, with $\Delta>2$ and positive multipliers. Otherwise it is parametric: the stiffness is restoring on average, and the growth is produced by its periodic modulation along the orbit. When $\Delta<-2$, as in the principal resonance of the Mathieu equation, the growing solution reverses its sign at every period and oscillates at half the orbital frequency: this is the subharmonic parametric resonance found in isentropic fluids and in the tongue issuing from $\nu_0=1/2$, whereas the tongue issuing from $\nu_0=1$, with $\Delta>2$, is a harmonic resonance. The sign of $\bar Q$ alone does not decide stability, since $\bar Q>0$ is compatible with instability (Mathieu equation) and $\bar Q<0$ with stability (Kapitza's inverted pendulum). In particular, the lower edge $\e_a$ of the stable band is where the mean restoring effect of the wave, $4\e^4/\cD^3$, balances the negative buoyancy stiffness $\nu^2/\cD$: there the orbital motion stabilises the convective mode dynamically, as in Kapitza's pendulum.

The results are sufficient conditions for instability. Only the wavevector $\be_2$ has been used; in-plane and oblique wavevectors may enlarge the unstable set, and were not examined in \citet{P} either. Sheared transverse currents (Section~\ref{sec:currents}), the vortex sheet of the discontinuous-current solution, and diabatic effects are outside the present analysis. As in \citet{P}, the instability tilts and stretches the two-dimensional trochoidal flow and is a natural seed for three-dimensional breakdown; for edge waves it is concentrated at the shoreline, and for surface waves in a layer beneath the free surface.

 \vskip 5mm
\noindent{\bf Acknowledgements.} The author used Claude Opus 5.5 (Anthropic) for the digitisation of the published profiles of \citet{H02} as well as to write and check the Python code for the subsequent numerical computations. The author checked all results and takes full responsibility for the content of the paper.

\vskip 2mm
\noindent{\bf Declaration of interest.}  The author declares no conflict of interest.   
 
 \vskip 2mm
\noindent{\bf Data availability statement.} {The data and code that support the findings of this study are openly available on Zenodo at \url{https://doi.org/10.5281/zenodo.23661634}. The deposit contains the temperature and wind profiles digitised from figure~3 of \citet{H02}, the other values taken from that article with their location in it, and the Python code that computes the stability diagram, the profiles of figure~\ref{fig:hecht} and every number quoted in the paper. The digitisation script requires a copy of the published article of \citet{H02}, which is not redistributed.}
\vskip 2mm
\noindent{\bf Author ORCiD.}  https://orcid.org/0009-0008-5454-0922

\vskip 2mm
\noindent{\bf Funding}This research received no specific grant from any funding agency, commercial or
not-for-profit sectors.

 \appendix
\section{Computation of the stability diagram}\label{app:numerics}
The discriminant $\hat\Delta(\e,\nu_0^2):=\Delta\bigl(\e,(1-\e^2)\nu_0^2\bigr)$ was computed by integrating over one period the Hill equation \eqref{eq:hill}, written as the first-order system $\mathcal A'=h\mathcal B$, $\mathcal B'=-Q\mathcal A$, whose discriminant is $\Delta$; integrating directly the three-dimensional system \eqref{eq:sysP} gives the same values to $10^{-10}$. We used the classical fourth-order Runge--Kutta scheme with 3000 steps per period, on a grid of $360\times381$ points covering $0.002\le\e\le0.95$ and $-0.4\le\nu_0^2\le1.5$. figure~\ref{fig:diagram}(\textit{a}) shows the growth rate $\sigma/\Omega=\operatorname{arcosh}(|\hat\Delta|/2)/(2\pi)$ where $|\hat\Delta|>2$, and the contour $|\hat\Delta|=2$; figure~\ref{fig:diagram}(\textit{b}) shows, for each $\nu_0^2\ge0$, the smallest steepness of the grid at which $|\hat\Delta|>2$, against $\lambda/H=2\pi\nu_0^2/\kappa$. The values quoted in Section~\ref{sec:numerics} were computed with 4000--8000 steps per period, on grids in $\e$ of spacing at most $2.5\times10^{-4}$ or by bisection on $\hat\Delta\pm2=0$. The critical steepness of figure~\ref{fig:hecht}(\textit{c}) was interpolated in a table of $\hat\Delta$ computed in the same way (Appendix~\ref{app:hecht}).

The scheme reproduces the exact discriminant \eqref{eq:Delta0} of neutral levels to $10^{-12}$, and centred finite differences reproduce the derivatives \eqref{eq:derivs} at $(1/3,0)$ to eight digits. As an independent check of the reduction of Section~\ref{sec:reduction}, the unreduced system \eqref{eq:wkb} was integrated in physical time along the trajectories \eqref{eq:paths}, for surface and edge waves, with uniform currents and with isothermal and polytropic reference states: the trace of its monodromy matrix, minus one, agrees with $\hat\Delta$ to about $10^{-11}$. The code is available as described in the data availability statement.

\section{Data of \texorpdfstring{\citet{H02}}{Hecht et al. (2002)} and derived profiles}\label{app:hecht}

\subsection{Data taken from \texorpdfstring{\citet{H02}}{Hecht et al. (2002)}}
Section~\ref{sec:physics} uses only published material from \citet{H02}. From the text we take the observed horizontal wavelength $\lambda=35\pm2$\,km, period $215\pm10$\,s and phase speed $c=160\pm10$\,m\,s$^{-1}$ (their \S\,3.1), the direction of propagation, towards $30^\circ$ east of south (\S\,3.2), the amplitude of the OH airglow fluctuation, about $2\,\%$ of the mean (\S\,3.1), the altitude of the OH emission layer, about 10\,km thick and centred near 87\,km (\S\,4.1.2), and the time during which the wave was observed, 09:54--10:12\, universal time (UT) on 17 November 1999 (\S\,3.1). From their figure~3 we take the Na lidar profiles of temperature and of the wind component along the direction of propagation between 80 and 108\,km, measured between 9.86 and 10.10\,UT and smoothed over 960\,m by the authors. From their figure~5 we read the buoyancy frequency and the intrinsic frequency of the wave at 85, 90, 105 and 107.5\,km, which are used only as a check. From their figures~6 and~8 we read the amplitudes of their full-wave solution, which enter only the discussion of the steepness. Each of these values is listed, with its location in \citet{H02}, in the data files that accompany this paper.

\subsection{Digitisation}
figure~3 of \citet{H02} is embedded in the published PDF file as a bitmap of $1804\times602$ pixels at 300 dots per inch, which we extracted without resampling. The axes were calibrated on the plot frames and the major tick marks, so that one pixel corresponds to 0.085\,km in altitude, 0.38\,K in temperature and 0.77\,m\,s$^{-1}$ in wind. The temperature curve, single-valued in altitude, was read row by row as the mean position of its pixels. In the wind panel, the solid line, which gives the component along the direction of propagation, was separated from the dashed zonal and meridional components as the largest connected set of pixels, and was followed upwards by choosing in each row the cluster of pixels that minimises the total squared horizontal displacement. Rows next to the frames were skipped, so that the tick marks are not taken for the curves. The temperature profile contains 318 points between 81.0 and 108.0\,km, and the wind profile 319 points between 81.0 and 108.1\,km. We estimate the error of the digitisation at about 1\,K and 2\,m\,s$^{-1}$, smaller than the uncertainty of the lidar data.

\subsection{Derived profiles}
The digitised profiles were interpolated linearly on a grid of spacing 0.05\,km over their common range, 81.0--108.0\,km. The local intrinsic frequency is $\omega_I=k(c-u)$, with $k=2\pi/\lambda$ and $u$ the wind along the direction of propagation, and the frequency of Godin's wave is $\Omega=(gk)^{1/2}$, with $g=g_0R_E^2/(R_E+z)^2$, $g_0=9.80665$\,m\,s$^{-2}$ and $R_E=6371$\,km. The buoyancy frequency of the reference state was computed as
\begin{equation}\label{eq:N0lidar}
N_0^2=\frac{g}{T}\left(\frac{\mathrm dT}{\mathrm dz}+\frac{g}{c_p}\right),\qquad c_p=1004.6\ \mathrm{J\,kg^{-1}\,K^{-1}},
\end{equation}
after smoothing the temperature with a Gaussian filter of full width at half maximum 1\,km or 2.5\,km; the second is the smoothing used by \citet{H02} for the background of their model. Then $\nu_0^2=N_0^2/\Omega^2$, since $\alpha=0$. The critical steepness at each altitude is the smallest $\e$ for which $|\Delta(\e,(1-\e^2)\nu_0^2)|>2$. It was obtained by interpolation in a table of the discriminant computed as in Appendix~\ref{app:numerics}, on 475 values of $\e$ in $[0.002,0.95]$ and 281 values of $\nu_0^2$ in $[-0.2,1.2]$, with 2000 Runge--Kutta steps per period, so that its resolution is 0.002.

The derived profiles reproduce the frequencies read from figure~5 of \citet{H02}. In that figure the intrinsic frequency is marked by a character X anchored at its lower-left corner (its centre lies at $k=1.86\times10^{-4}$\,m$^{-1}$ in all four panels, instead of $2\pi/\lambda=1.80\times10^{-4}$\,m$^{-1}$), and we read it at the bottom of the character. The intrinsic frequencies then agree with ours to within $3\,\%$ at the four altitudes. The values of $\nu_0^2$ obtained from their buoyancy frequencies differ from ours (2.5\,km smoothing) by 0.03 at 85\,km, 0.01 at 90\,km and 0.06 at 105\,km. At 107.5\,km, near the upper end of the lidar data, our value is lower, 0.62 against 0.92, and the derived profiles are therefore not used above 107.6\,km.

The amplitudes of the full-wave solution of \citet{H02} quoted in Section~\ref{sec:physics} come from their figure~8. Their energy density is $0.82\times10^{-4}$\,J\,m$^{-3}$ at 81\,km, where the ratio of potential to kinetic energy is 0.40, and it falls below the resolution of the figure, $0.03\times10^{-4}$\,J\,m$^{-3}$, above 95\,km; at 100\,km the ratio is 0.08. With their equation~(8) and the densities of the U.S. Standard Atmosphere \citep{USSA76}, this gives a velocity amplitude of about 3.9\,m\,s$^{-1}$ at 81\,km and below 4.5\,m\,s$^{-1}$ at 100\,km, that is a steepness $\e=|u|/(g/k)^{1/2}<0.02$ there.

\begin{thebibliography}{99}

\expandafter\ifx\csname natexlab\endcsname\relax
\def\natexlab#1{#1}\fi
\expandafter\ifx\csname selectlanguage\endcsname\relax
\def\selectlanguage#1{\relax}\fi




\bibitem[{Abdolali \& Kirby(2017)}]{AK17}
\textsc{Abdolali, A. \& Kirby, J.T.} 2017 Role of compressibility on tsunami propagation. \emph{J. Geophys. Res. Oceans} \textbf{122} (12), 9780--9794. DOI: 10.1002/2017JC013054
 
\bibitem[{Abdolali et~al.(2019)Abdolali, Kadri \& Kirby}]{AKK19}
\textsc{Abdolali, A., Kadri, U. \& Kirby, J.T.} 2019 Effect of water compressibility, sea floor elasticity, and field gravitational potential on tsunami phase speed. \emph{Sci. Rep.} \textbf{9}, 16874.
 
\bibitem[{Abdolali et~al.(2015)Abdolali, Kirby \& Bellotti}]{AKB15}
\textsc{Abdolali, A., Kirby, J.T. \& Bellotti, G.} 2015 Depth-integrated equation for hydro-acoustic waves with bottom damping. \emph{J. Fluid Mech.} \textbf{766}, R1. DOI: 10.1017/jfm.2015.37




\bibitem[{Babenko(1987)}]{Babenko87}
\textsc{Babenko, K.I.} 1987 Some remarks on the theory of surface waves of finite amplitude. \emph{Dokl. Akad. Nauk SSSR} \textbf{294}, 1033--1037.


\bibitem[Constantin(2001)]{C2001}
{\sc Constantin, A.} 2001 Edge waves along a sloping beach
{\em J. Phys. A: Math. Gen.}, {\bf 34}, 9723---9731. DOI:  10.1088/0305-4470/34/45/311


\bibitem[Constantin(2023)]{C23}
{\sc Constantin, A.} 2023 Exact nonlinear mountain waves propagating upwards
{\em J. Phys. A: Math. Theor.}, {\bf 56}, 245702. DOI: 10.1088/1751-8121/acd429


\bibitem[Constantin \& Germain(2013)]{CG}
{\sc Constantin, A. \& Germain, P.} 2013 Instability of some equatorially trapped waves {\em J. Geophys. Res. Oceans}, {\bf 118}, 2802---2810.
DOI: 10.1002/jgrc.20219

 
\bibitem[{Derakhti \& Kirby(2016)}]{DK16}
\textsc{Derakhti, M. \& Kirby, J.T.} 2016 Breaking-onset, energy and momentum flux in unsteady focused wave packets. \emph{J. Fluid Mech.} \textbf{790}, 553--581.
 
\bibitem[{Derakhti et~al.(2020)Derakhti, Kirby, Banner, Grilli \& Thomson}]{DKBGT20}
\textsc{Derakhti, M., Kirby, J.T., Banner, M.L., Grilli, S.T. \& Thomson, J.} 2020 A unified breaking onset criterion for surface gravity water waves in arbitrary depth. \emph{J. Geophys. Res. Oceans} \textbf{125}, e2019JC015886. DOI: 10.1029/2019JC015886



\bibitem[{Fritts \& Alexander(2003)}]{FA03}
\textsc{Fritts, D.C. \& Alexander, M.J.} 2003 Gravity wave dynamics and effects in the middle atmosphere. \emph{Rev. Geophys.} \textbf{41} (1), 1003. DOI: 10.1029/2001RG000106

\bibitem[{Gill(1982)}]{Gill82}
\textsc{Gill, A.E.} 1982 \emph{Atmosphere--Ocean Dynamics}. Academic Press. 
 
\bibitem[Godin(2015)]{G}
{\sc Godin, O.A.} 2015 Finite-amplitude acoustic-gravity waves: exact solutions
{\em J. Fluid Mech.}, {\bf 767}, 52---64. DOI: 10.1017/jfm.2015.40


\bibitem[{Godin et~al.(2015)Godin, Zabotin \& Bullett}]{GZB15}
\textsc{Godin, O.A., Zabotin, N.A. \& Bullett, T.W.} 2015 Acoustic-gravity waves in the atmosphere generated by infragravity waves in the ocean. \emph{Earth Planets Space} \textbf{67}, 47. DOI: 10.1186/s40623-015-0212-4
 
\bibitem[{Gossard \& Hooke(1975)}]{GH75}
\textsc{Gossard, E.E. \& Hooke, W.H.} 1975 \emph{Waves in the Atmosphere}. Elsevier.



\bibitem[Hecht et al.(2002)]{H02}
{\sc Hecht, J.H., Walterscheid, R.L., Hickey, M.P., Rudy, R.J. \& Liu, A.Z.} 2002
An observation of a fast external atmospheric acoustic-gravity wave
{\em J. Geophys. Res.}, {\bf 107}, 4444. DOI: 1029/2001JD001438


 
\bibitem[{Hines(1960)}]{H60}
\textsc{Hines, C.O.} 1960 Internal atmospheric gravity waves at ionospheric heights. \emph{Can. J. Phys.} \textbf{38}, 1441--1481. DOI: 10.1139/p60-150



\bibitem[Ionescu-Kruse(2014)]{IK14}
{\sc Ionescu-Kruse, D.} 2014 Instability of edge waves along a sloping beach
{\em J. Differ. Equ.}, {\bf 256}, 3999---4012


\bibitem[{Kirby \& Chen(1989)}]{KC89}
\textsc{Kirby, J.T. \& Chen, T.-M.} 1989 Surface waves on vertically sheared flows: approximate dispersion relations. \emph{J. Geophys. Res.} \textbf{94} (C1), 1013--1027. DOI: 10.1029/JC094iC01p01013
 
  \bibitem[{Kirby et~al.(1981)Kirby, Dalrymple \& Liu}]{KDL81}
  \textsc{Kirby, J.T., Dalrymple, R.A. \& Liu, P.L.-F.} 1981 Modification of edge waves by barred-beach topography. \emph{Coastal Engng} \textbf{5}, 35--49.
 


\bibitem[Leblanc(2004)]{L04}
{\sc Leblanc, S.} 2004 Local stability of Gerstner's waves. {\em J. Fluid Mech.}, {\bf 506}, 245---254. DOI: 10.1017/S0022112004008444



\bibitem[Lebovitz \& Lifschitz(1992)]{LL92}
{\sc Lebovitz, N. \& Lifschitz, A.} 1992 Short wavelength instabilities of rotating compressible fluid masses.
{\em Proc. R. Soc. Lond. A}, {\bf 438}, 265---290. DOI: 10.1098/rspa.1992.0106

\bibitem[Lifschitz \& Hameiri(1991)]{LH91}
{\sc Lifschitz, A. \& Hameiri, E.} 1991 Local stability conditions in fluid dynamics.
{\em Phys. Fluids A}, {\bf 3}, 2644---2651. DOI: 10.1063/1.858153


\bibitem[NOAA, NASA \& USAF(1976)]{USSA76}
{\sc NOAA, NASA \& USAF} 1976
\newblock {\em U.S. Standard Atmosphere, 1976}.
\newblock Washington, DC: U.S. Government Printing Office.

\bibitem[Puntini(2026)]{P26}
{\sc Puntini, C.} 2026 Instability of the Halocline at the North Pole. {\em J. Phys. Oceanogr.}, {\bf 56}, 1499---1511. DOI: 10.1175/JPO-D-26-0029.1


\bibitem[Puntini(2027)]{P}
{\sc Puntini, C.} 2027 On the instability of some upward propagating, exact, nonlinear mountain waves. {\em Nonlinear Anal.: Real World Appl.}, {\bf 95}, 104781. DOI: 10.1016/j.nonrwa.2026.104781

\bibitem[{She et~al.(2000)She, Chen, Hu, Sherman, Vance, Vasoli, White, Yu \& Krueger}]{She00}
\textsc{She, C.Y., Chen, S., Hu, Z., Sherman, J., Vance, J.D., Vasoli, V., White, M.A., Yu, J. \& Krueger, D.A.} 2000 Eight-year climatology of nocturnal temperature and sodium density in the mesopause region (80 to 105\,km) over Fort Collins, CO (41$^\circ$N, 105$^\circ$W). \emph{Geophys. Res. Lett.} \textbf{27}. DOI: 10.1029/2000GL003825


 \bibitem[{Stokes(1880)}]{Stokes1880}
\textsc{Stokes, G.G.} 1880 Supplement to a paper on the theory of oscillatory waves. In \emph{Mathematical and Physical Papers}, vol.~1, pp.~314--326. Cambridge University Press.


\bibitem[{Toffoli et~al.(2010)Toffoli, Babanin, Onorato \& Waseda}]{TBOW10}
\textsc{Toffoli, A., Babanin, A., Onorato, M. \& Waseda, T.} 2010 Maximum steepness of oceanic waves: field and laboratory experiments. \emph{Geophys. Res. Lett.} \textbf{37}, L05603. DOI: 10.1029/2009GL041771




\bibitem[{Williams(1981)}]{Williams1981}
\textsc{Williams, J.M.} 1981 Limiting gravity waves in water of finite depth. \emph{Phil. Trans. R. Soc. Lond.} A \textbf{302}, 139--188.



\bibitem[Yih(1966)]{Yih}
{\sc Yih, C.-S.} 1966 Note on edge waves in a stratified fluid {\em J. Fluid Mech.}, {\bf 24}, 765-767





 \end{thebibliography}
\end{document}